\documentclass[journal, onecolumn]{IEEEtran}
\usepackage{ amsfonts, amsthm, MnSymbol, algorithm, algpseudocode, graphicx, enumerate, mathtools, bm, nicematrix, bbm}
\usepackage{xcolor}
\usepackage[normalem]{ulem} % for strike through

\theoremstyle{definition}
\newtheorem{definition}{Definition}%[section]
\newtheorem{remark}{Remark}%[definition]
\newtheorem{theorem}{Theorem}%[section]
\newtheorem{corollary}{Corollary}%[theorem]
\newtheorem{lemma}{Lemma}%[theorem]
\newtheorem{prop}{Proposition}%[theorem]
\newtheorem{example}{Example}%[section]
\usepackage{graphicx}
\usepackage[colorlinks=true, allcolors=blue]{hyperref}

\newcommand{\tgoppa}{\Gamma(\mathcal{L},g,\bm{\mathbbm{t}},\bm{\mathbbm{h}},\bm{\mathbbm{\eta}})}

\begin{document}
\title{Cryptographic Applications of Twisted Goppa Codes}
\author{Harshdeep Singh$^{1,2}$, Anuj Kumar Bhagat$^{1}$, Ritumoni Sarma$^1$ and Indivar Gupta$^2$% <-this % stops a space

\thanks{*This work was partially supported by DRDO, Government of India and CSIR-HRDG}% <-this % stops a space
\thanks{$^{1}$Department of Mathematics, Indian Institute of Technology Delhi, Hauz Khas, Delhi-110016, India
       % {\tt\small harshdeep.sag@gov.in}
       }%
\thanks{$^{2}$Scientific Analysis Group, Defence R\&D Organisation, Metcalfe House, Delhi-110054, India
       % {\tt\small p.misra at ieee.org}
       }%
\thanks{Dr. Ritumoni Sarma is Professor at Department of Mathematics, Indian Institute of Technology Delhi}
}

\markboth{Journal of ...}%
{Shell \MakeLowercase{\textit{et al.}}: \title}

\maketitle
%\thispagestyle{empty}
%\pagestyle{empty}

%%%%%%%%%%%%%%%%%%%%%%%%%%%%%%%%%%%%%%%%%%%%%%%%%%%%%%%%%%%%%%%%%%%%%%%%%%%%%%%%
\begin{abstract}
This article defines multi-twisted Goppa (MTG) codes as subfield subcodes of duals of multi-twisted Reed-Solomon (MTRS) codes and examines their properties. 
We show that if $t$ is the degree of the MTG polynomial defining an MTG code, its minimum distance is at least $t + 1$ under certain conditions. 
Extending earlier methods limited to single twist at last position, we use the extended Euclidean algorithm to efficiently decode MTG codes with a single twist at any position, correcting up to $\left\lfloor \tfrac{t}{2} \right\rfloor$ errors.
This decoding method highlights the practical potential of these codes within the Niederreiter public key cryptosystem (PKC). 
Furthermore, we establish that the  Niederreiter PKC based on MTG codes is secure against partial key recovery attacks. 
Additionally, we also reduce the public key size by constructing quasi-cyclic MTG codes using a non-trivial automorphism group.
\end{abstract}
% \tableofcontents

%%%%%%%%%%%%%%%%%%%%%%%%%%%%%%%%%%%%%%%%%%%%%%%%%%%%%%%%%%%%%%%%%%%%%%%%%%%%%%%%
\section{INTRODUCTION}
Reed-Solomon (RS) codes are derived from the evaluation of certain polynomials at distinct points of a finite field. 
Twisted Reed-Solomon codes (TRS) are the generalizations of RS codes by adding a single twisted term, first introduced by Beelen et al. in \cite{Beelen2017}. 
Subsequent contributions by Beelen et al. in \cite{Beelen2022} further expanded this framework by introducing TRS with more than one twist; 
such TRS codes are now referred to as multi-twisted Reed-Solomon (MTRS) codes in the literature. 
The authors in \cite{Beelen2022} also produced classes of MTRS codes, which are MDS and are not equivalent to a generalized Reed-Solomon code. 
Additionally, Beelen et al. in \cite{BeelenStructural2018} showed that TRS codes have good structural properties to be used in cryptography and in \cite{Beelen2022} gave a description of a decoding algorithm for decoding MTRS codes. 
Hence, MTRS codes can replace Goppa codes in the McEliece PKC. 
However, Lavauzelle et al. in \cite{Lavauzelle2020} gave an efficient key recovery algorithm for a McEliece PKC based on a TRS code using subfield subcodes of the TRS code.

Twisted Goppa codes were recently introduced by Sui and Yue as subfield subcodes of a TRS code~\cite{Sui2023}, who also showed that these codes exhibit properties similar to those of classical Goppa codes. 
Additionally, they proved that under certain conditions, the minimum distance of a twisted Goppa code is at least $t+1$ and gave an efficient decoding algorithm for twisted Goppa codes that corrects $\lfloor \frac{t-1}{2}\rfloor$ errors, where $t$ denotes the degree of Goppa polynomial. 
Consequently, twisted Goppa codes can also serve as suitable replacements for classical Goppa codes in the McEliece public-key cryptosystem. 
With a code having minimum distance $d \ge t + 1$, up to $\lfloor \tfrac{t}{2} \rfloor$ errors can be corrected. 
In \cite{Sun2025}, Sun et al.\ proposed an efficient decoding algorithm that achieves this bound for twisted Goppa codes. 
Both \cite{Sui2023} and \cite{Sun2025} adopt a particular formulation for twisted Goppa codes, in which the twist is introduced at the last position considering parity check matrix of the code, inherently limiting to particular class of codes. 
Motivated by this, we extend the construction to the subfield subcodes of duals of multi-twisted Reed-Solomon (MTRS) codes, referred to as \emph{multi-twisted Goppa codes} in this work.
These codes generalize the notion of twisted Goppa codes with no restriction on the number of twists and their positions.

%Twisted Goppa codes were recently introduced by Sui and Yue \cite{c1}, in which they provide an efficient decoding algorithm. This enables one to study the cryptographic applications of these codes. The concise definition used by authors involve a single twist in Goppa codes at a fixed position, this provides bound on error correction capacity utilizing the work of Peter Beelen \cite{p-beelen-2018}. However, the definition is now more generalised in our work, with no restriction on number of twists and their positions. Utilizing the foundations laid in \cite{h-k-2023}, we provide necessary and sufficient conditions on the error correction capacity of Goppa codes. 
%Additionally, we provide the error correction algorithm for the new family of codes. 

%The multivariate Goppa codes were also recently introduced by L\`{o}pez et al. in \cite{lopez-2023}, by considering them as subfield subcodes of Augmented Cartesian codes. They studied the properties of duality and dimension of their hulls, enabling these codes to be used in entanglement-assisted quantum error-correction.

%The concept of extending single variable Goppa codes to multiple variables paves the way for extending our constructed codes to multivariate. Thus, our work improves and generalises the previous work by \cite{}.
In this article, we define a multi-twisted Goppa (MTG) code and prove that it is a subfield subcode of dual of an MTRS code. 
We show that the properties of an MTG code are similar to the properties of a single twist Goppa code and hence similar to the properties of a classical Goppa code. 
Utilizing the foundations laid in \cite{h-k-2023}, we provide necessary and sufficient conditions on the error correction capacity of Goppa codes. 
We provide an efficient decoding algorithm for codes whose dual is an MDS MTRS code and consequently provide an efficient algorithm to decode twisted Goppa codes (with no restriction on twist position). 
As a result, MTG codes can be employed in Niederreiter/ McEliece PKC.

The remainder of this article is organized as follows. The next section presents all the necessary preliminaries, including the definitions and properties of multi-twisted Goppa codes. 
In Section~\ref{Section: Twisted Goppa Codes}, we develop a decoding algorithm based on the Extended Euclidean Algorithm for twisted Goppa codes with a twist at an arbitrary position and analyze its computational complexity. 
Section~\ref{Section: Cryptosystem based on Twisted Goppa codes} introduces a cryptosystem based on twisted Goppa codes. 
Furthermore, in the same section, we demonstrate that the cryptosystem constructed from multi-twisted Goppa codes is likely to be secure against best-known partial key-recovery attacks and present quasi-cyclic twisted Goppa codes to reduce the size of the public key. 

\section{PRELIMINARIES}
Let $\mathbb{F}_q$ be the finite field of size $q$. Let $k<q$ be a positive integer.
Considering $\bm{\mathbbm{t}} = (t_1,t_2,\dots,t_{\ell})$ with $1\leq t_i \leq q-k$, $\bm{\mathbbm{h}}=(h_1,h_2,\dots,h_{\ell})$ with $0\leq h_i<h_{i+1}<k$ for each $i$, and $\bm\eta= (\eta_1,\eta_2,\dots, \eta_{\ell}) \in (\mathbb{F}_q^{\ast})^\ell$, set of $(k,\bm{\mathbbm{t}}, \bm{\mathbbm{h}}, \bm{\eta})$-twisted polynomials over $\mathbb{F}_q$, is defined as
\begin{equation*}
	\mathcal{V}_{k, \bm{\mathbbm{t}}, \bm{\mathbbm{h}}, \bm{\eta}}:= \left \{\sum\limits_{i=0}^{k-1}a_i x^i + \sum\limits_{j=1}^\ell \eta_j a_{h_j} x^{k-1+t_j}: a_i \in \mathbb{F}_q \text{ for each }i \right \}.
\end{equation*}
\begin{definition}
        [Multi-twisted Reed Solomon code]\label{def::mtrs}
        \cite{Beelen2022}
    	Let $\alpha_1, \alpha_2, \ldots, \alpha_n$ denote distinct elements of the finite field $\mathbb{F}_q$. Let the parameters $k, \bm{\mathbbm{t}}, \bm{\mathbbm{h}}, \bm{\eta}$ be defined as above, with the additional constraint that $k < n\le q$.
    The \emph{multi-twisted Reed-Solomon (MTRS) code} is defined as
    \begin{equation}\label{eqn3}
        \mathcal{C}_k(\bm{\alpha}, \bm{\mathbbm{t}}, \bm{\mathbbm{h}}, \bm{\eta})
        := \{(f(\alpha_1), f(\alpha_2), \dots, f(\alpha_n)) : f(x) \in \mathcal{V}_{k, \bm{\mathbbm{t}}, \bm{\mathbbm{h}}, \bm{\eta}}\},
    \end{equation}
    where $\deg(f(x)) \leq k - 1 + t_\ell < n$. 
    This construction generalizes the classical Reed-Solomon code by incorporating multiple twists determined by the tuples 
    $\bm{\mathbbm{t}}$, $\bm{\mathbbm{h}}$, and $\bm{\eta}$.
    Furthermore, by introducing a nonzero \emph{multiplier vector} $\bm{v} = (v_1, v_2, \dots, v_n) \in (\mathbb{F}_q^*)^n$, 
    we obtain the \emph{generalized multi-twisted Reed-Solomon (MT-GRS) code}, defined as
    \begin{equation}\label{eqn::MT GRS}
        \mathcal{C}_k(\bm{\alpha}, \bm{\mathbbm{t}}, \bm{\mathbbm{h}}, \bm{\eta}, \bm{v})
        := \{(v_1 f(\alpha_1), v_2 f(\alpha_2), \dots, v_n f(\alpha_n)) : f(x) \in \mathcal{V}_{k, \bm{\mathbbm{t}}, \bm{\mathbbm{h}}, \bm{\eta}}\}.
    \end{equation}
\end{definition}
The inclusion of $\bm{v}$ allows a scaling of codeword coordinates, thereby extending the MTRS family to a broader class of codes that preserve linearity while enhancing structural flexibility. The next result provides a necessary and sufficient conditions for these codes to be MDS.

\begin{theorem}\cite{h-k-2023}\label{lemma}
	Let $\alpha_1, \alpha_2, \dots, \alpha_n \in \mathbb{F}_q$ be distinct and $k<n<q$ and $\bm{\mathbbm{t}}=(t_1,t_2,\dots,t_\ell)$, $\bm{\mathbbm{h}}=(h_1,h_2,\dots,h_\ell)$ and $\bm \eta=(\eta_1,\eta_2,\dots, \eta_\ell)$ be as per Definition \ref{def::mtrs}.
	Then the multi-twisted RS code $\mathcal{C}_{k}(\bm\alpha, \bm{\mathbbm{t}},\bm{\mathbbm{h}},\bm \eta)$ is MDS if and only if 
	for each $\mathcal{I} \subset \{1,2,\dots ,n\}$ with cardinality $k$ correspondingly the polynomial $\prod\limits_{i\in \mathcal{I}}(x-\alpha_i)= \sum\limits_{i} \sigma_i x^i$ with $\sigma_i =0$ for $i<0$,	
	the matrix
	\begin{equation*}
		\text{diag }\Big( \underset{\underset{t_{\ell}^{th}}{\uparrow} }{\eta_{\ell}^{-1}}, 1,\dots, 1, \underset{\underset{t_{\ell-1}^{th}}{\uparrow} }{\eta_{\ell-1}^{-1}}, 1,\dots, 1, \underset{\underset{t_{1}^{th}}{\uparrow} }{\eta_{1}^{-1}}, 1,\dots, \underset{\underset{1^{st}}{\uparrow} }{1} \Big)
		\cdot A_{\mathcal{I}} +B_\mathcal{I}
	\end{equation*}
	is non-singular. 
	Here, $A_\mathcal{I}$ and $B_\mathcal{I}$ are lower and upper-triangular $t_\ell\times t_\ell$ matrices, respectively, given by: 
	\begin{equation*}
		A_\mathcal{I} = \left(\begin{array}{lllll}
			1& 0 &0 & \cdots & 0\\
			\sigma_{k-1}& 1 & 0 &\cdots &0\\
			\sigma_{k-2}& \sigma_{k-1} & 1 &\cdots &0\\
			~~\vdots&~\vdots&~\vdots&~~\ddots&~\vdots\\
			\sigma_{k - t_\ell +1} & \cdots &\cdots & \sigma_{k-1}&1\\
		\end{array}\right) \text{ and}
	\end{equation*}
		% \scalebox{0.6}{\begin{equation*}
		% 	B_\mathcal{I} = \left(\begin{array}{ccccccccc}
		% 		-\sigma_{h_\ell-t_\ell+1}& -\sigma_{h_\ell-t_\ell+2} &\cdots & \cdots & \cdots & \cdots&\cdots&\cdots&\cdots\\
		% 		0&\cdots&\cdots&\cdots&\cdots&\cdots&\cdots&\cdots&0\\
				
		% 		\vdots&\vdots&\vdots&\vdots&\vdots&\vdots&\vdots&\vdots&\vdots\\
				
		% 		0&\cdots&\cdots&\cdots&\cdots&\cdots&\cdots&\cdots&0\\
		% 		0& \cdots&0&-\sigma_{h_{\ell-1}-t_{\ell-1}+1} & -\sigma_{h_{\ell-1}-t_{\ell-1}+2} &\cdots&\cdots&\cdots&\cdots \\
				
		% 		0&\cdots&\cdots&0&\cdots&\cdots&\cdots&\cdots&0\\
				
		% 		\vdots&\vdots&\vdots&\vdots&\vdots&\vdots&\vdots&\vdots&\vdots\\
				
		% 		0&\cdots&\cdots&\cdots&\cdots&\cdots&\cdots&\cdots&0\\
		% 		0& \cdots&\cdots&\cdots& 0 &-\sigma_{h_{1}-t_{1}+1} & -\sigma_{h_{1}-t_{1}+2} &\cdots&\cdots \\
		% 		0&\cdots&\cdots&\cdots&\cdots&0&\cdots&\cdots&0\\
				
		% 		\vdots&\vdots&\vdots&\vdots&\vdots&\vdots&\vdots&\vdots&\vdots\\
				
		% 		0&\cdots&\cdots&0&\cdots&0&\cdots&\cdots&0\\
		% 	\end{array}\right)
		% 	\begin{array}{l}
		% 		\leftarrow t_\ell^{th}\\
		% 		\\
		% 		\\
		% 		\\
		% 		\leftarrow t_{\ell-1}^{th}
		% 		\\
		% 		\\
		% 		~~~\vdots\\
		% 		\\
		% 		\leftarrow t_1^{th}\\
		% 		\\
		% 		\\
		% 		\leftarrow 1^{st}\\
		% 	\end{array}
		% \end{equation*}
		% \begin{equation*}
		% 	\begin{array}{ccccccccc}
		% 		\uparrow~~~~~~~~~~~&&&~~~~~~~~~~~\uparrow~~~~~~~~~~~~~&&~~~~~~~~~~~\uparrow~~~~~~~~~~&&&~~~~~~~\uparrow\\
		% 		t_\ell^{th}~~~~~~~~~~~~&&&~~~~~~~~~~~t_{\ell-1}^{th}~~~~~~~~~~~~~&\cdots&~~~~~~~~~~~t_1^{th}~~~~~~~~~~&&&~~~~~~~1^{st}\\
		% 	\end{array}
		% \end{equation*}}
        \begin{equation*}
\displaystyle
B_\mathcal{I} = \begin{pNiceMatrix}[last-row, last-col, nullify-dots]
	-\sigma_{h_\ell-t_\ell+1}& -\sigma_{h_\ell-t_\ell+2} &\cdots & \cdots & \cdots & \cdots&\cdots&\cdots&\cdots&\leftarrow t_\ell^{\text{th}}\\
	0&\cdots&\cdots&\cdots&\cdots&\cdots&\cdots&\cdots&0\\
	\vdots&\vdots&\vdots&\vdots&\vdots&\vdots&\vdots&\vdots&\vdots\\
	0&\cdots&\cdots&\cdots&\cdots&\cdots&\cdots&\cdots&0\\
	0& \cdots&0&-\sigma_{h_{\ell-1}-t_{\ell-1}+1} & -\sigma_{h_{\ell-1}-t_{\ell-1}+2} &\cdots&\cdots&\cdots&\cdots \leftarrow t_{\ell-1}^{\text{th}}\\
	0&\cdots&\cdots&0&\cdots&\cdots&\cdots&\cdots&0\\
	\vdots&\vdots&\vdots&\vdots&\vdots&\vdots&\vdots&\vdots&\vdots\\
	0&\cdots&\cdots&\cdots&\cdots&\cdots&\cdots&\cdots&0\\
	0& \cdots&\cdots&\cdots& 0 &-\sigma_{h_{1}-t_{1}+1} & -\sigma_{h_{1}-t_{1}+2} &\cdots&\cdots&\leftarrow t_1^{\text{th}} \\
	0&\cdots&\cdots&\cdots&\cdots&0&\cdots&\cdots&0\\
	\vdots&\vdots&\vdots&\vdots&\vdots&\vdots&\vdots&\vdots&\vdots\\
	0&\cdots&\cdots&0&\cdots&0&\cdots&\cdots&0&\leftarrow 1^{\text{st}}\\
     \uparrow&&&\uparrow&&\uparrow&&&\uparrow&\\
        t_\ell^{\text{th}}&&&t_{\ell-1}^{\text{th}}&\cdots&t_1^{\text{th}}&&&1^{\text{st}}\\
\end{pNiceMatrix}
% \begin{array}{l}
% 	\leftarrow t_\ell^{\text{th}}\\
% 	\\
% 	\\
% 	\\
% 	\leftarrow t_{\ell-1}^{\text{th}}\\
% 	\\
% 	~~~\vdots\\
% 	\\
% 	\leftarrow t_1^{\text{th}}\\
% 	\\
% 	\\
% 	\leftarrow 1^{\text{st}}\\
% \end{array}
\end{equation*}
% \begin{equation*}\resizebox{0.5\linewidth}{!}{%
% $\displaystyle
% \begin{array}{ccccccccc}
% 	~~~~~~~~~~~~~~~\uparrow~~~~~~~~~~~~~~~~&&&~~~~~~\uparrow~~~~~~~~~~~~~~&&~\uparrow~~~~~~~~&&&~~~~\uparrow~~~~~~~\\
% 	~~~~~~~~~~~~~t_\ell^{\text{th}}~~~~~~~~~~~~&&&~~~~t_{\ell-1}^{\text{th}}~~~~~~~~~~~~~&\cdots&~~~~~~~t_1^{\text{th}}~~~~~~~~~~&&&~~~1^{\text{st}}\\
% \end{array}
% $}
% \end{equation*}

\end{theorem}

Let $\mathbb{F}_{q^m}$ be the finite field of size $q^m$, for some prime power $q$ and $m$ a positive integer.
Let $g(z)=g_0+g_1z+\cdots+g_tz^t\in\mathbb{F}_{q^m}[z]$ with deg $g(z)=t$. % for $t$ being some positive integer.
We define an $\mathbb{F}_{q^m}[z]$-module transformation $\operatorname{U}:{\mathbb{F}_{q^m}[z]}^{t+1}\longrightarrow {\mathbb{F}_{q^m}[z]}^{t+1}$ as $\operatorname{U}(\bm{e}_i)=\bm{e}_{i+1}$, for $1\leq i\leq t$ and $\operatorname{U}(\bm{e}_{t+1})=\bm{0}$, where $\bm{e}_i$ denotes the unit vector with only non-zero coordinate at the $i^{th}$ place being $1$.
Throughout this paper, we denote $\bm{z}=(1,z,z^2,\dots,z^t)$ and $\bm{g}=(g_0,g_1,g_2,\dots,g_t)\in {\mathbb{F}_{q^m}[z]}^{t+1}$. Notice that $\bm{z}\cdot \bm{g}=g(z)$, additionally, $\operatorname{U}^j(\bm{z})\cdot \bm{g}$ can be easily computed for any $j\in\mathbb{N}$.

\begin{definition}
	[Multi-twisted Goppa codes]\label{def::MTGC}
	Let $g(z)=g_0+g_1z +\cdots +g_tz^t\in\mathbb{F}_{q^m}[z]$ with $g_t\neq0$ and $\mathcal{L}=\{\alpha_1,\alpha_2,\dots,\alpha_n\}\subseteq \mathbb{F}_{q^m}$ be such that $g(\alpha_i)\neq 0$ for each $i$. Let $\ell\geq0$ be an integer, and $\bm{\mathbbm{t}}=(t_1,t_2,\dots,t_\ell)\in\mathbb{Z}^\ell$, ${\bm{\mathbbm{h}}}=(h_1,h_2,\dots,h_\ell)$ be such that $1\leq t_i< t_{i+1} <q^m -t$ and $0\leq h_i\leq h_{i+1}< t$ for each $i$ and $\bm{\eta}= (\eta_1,\eta_2,\dots,\eta_\ell)\in \mathbb{F}_{q^m}^\ell$. Then a multi-twisted Goppa code $\tgoppa$ is defined as
    %\begin{multline}
        		%\bigg\{\bm{c}\in \mathbb{F}_q^n :\sum\limits_{i=1}^n c_i \left( \frac{1}{z-\alpha_i} - \sum\limits_{j=1}^\ell \frac{\eta_j \alpha_i^{t-1+t_j}}{g(\alpha_i)} \left(\operatorname{U}^{h_j+1}(\bm{z})\cdot \bm{g}\right) \right)\\
                %\equiv 0\bmod g(z) \bigg\},
    %\end{multline}
    $\bigg\{\bm{c}=(c_1, c_2, \dots, c_n)\in \mathbb{F}_q^n :\sum\limits_{i=1}^n c_i \left( \frac{1}{z-\alpha_i} - \sum\limits_{j=1}^\ell \frac{\eta_j \alpha_i^{t-1+t_j}}{g(\alpha_i)} \left(\operatorname{U}^{h_j+1}(\bm{z})\cdot \bm{g}\right) \right)\equiv 0\bmod g(z) \bigg\}$.
\end{definition}
\begin{prop}\label{prop::parity_mat_mtgoppa}
	Considering the above definition, $\tgoppa$ is Null$(H)\cap \mathbb{F}_q^n$, where $H\in\mathbb{F}_{q^m}^{t\times n}$ is given by 
	
        \begin{equation}\label{pc-mat-MTG}
		%\resizebox{0.5\linewidth}{!}{%
        %$
                \left(\begin{array}{cccc}
			g(\alpha_1)^{-1} & g(\alpha_2)^{-1} & \cdots & g(\alpha_n)^{-1}\\
			\alpha_1g(\alpha_1)^{-1} & \alpha_2g(\alpha_2)^{-1} & \cdots & \alpha_ng(\alpha_n)^{-1} \\
			\vdots & \vdots & \cdots & \vdots \\
			(\alpha_1^{h_1}+ \eta_1 \alpha_1^{t-1+t_1})g(\alpha_1)^{-1} & (\alpha_2^{h_1}+ \eta_1 \alpha_2^{t-1+t_1})g(\alpha_2)^{-1} & \cdots & (\alpha_n^{h_1}+ \eta_1 \alpha_n^{t-1+t_1})g(\alpha_n)^{-1} \\
			\alpha_1^{h_1+1}g(\alpha_1)^{-1}&\alpha_2^{h_1 +1}g(\alpha_2)^{-1}&\cdots&\alpha_n^{h_1 +1}g(\alpha_n)^{-1}\\
			\vdots &\vdots &\vdots &\vdots \\
			(\alpha_1^{h_2}+ \eta_2 \alpha_1^{t-1+t_2})g(\alpha_1)^{-1} & (\alpha_2^{h_2}+ \eta_2 \alpha_2^{t-1+t_2})g(\alpha_2)^{-1} & \cdots & (\alpha_n^{h_2}+ \eta_2 \alpha_n^{t-1+t_2})g(\alpha_n)^{-1}\\
			\alpha_1^{h_2+1}g(\alpha_1)^{-1}&\alpha_2^{h_2 +1}g(\alpha_2)^{-1}&\cdots&\alpha_n^{h_2 +1}g(\alpha_n)^{-1}\\
			\vdots &\vdots &\vdots &\vdots \\
			(\alpha_1^{h_{\ell}}+ \eta_{\ell} \alpha_1^{t-1+t_{\ell}})g(\alpha_1)^{-1} & (\alpha_2^{h_{\ell}}+ \eta_{\ell} \alpha_2^{t-1+t_{\ell}})g(\alpha_2)^{-1} & \cdots & (\alpha_n^{h_{\ell}}+ \eta_{\ell} \alpha_n^{t-1+t_{\ell}})g(\alpha_n)^{-1}\\
			\alpha_1^{h_{\ell}+1}g(\alpha_1)^{-1}&\alpha_2^{h_{\ell} +1}g(\alpha_2)^{-1}&\cdots&\alpha_n^{h_{\ell} +1}g(\alpha_n)^{-1}\\
			\vdots &\vdots &\vdots &\vdots \\
			\alpha_1^{t-1}g(\alpha_1)^{-1} & \alpha_2^{t-1}g(\alpha_2)^{-1} & \cdots &  \alpha_n^{t-1}g(\alpha_n)^{-1} \\
		\end{array}\right)_{t \times n}\begin{array}{c}
			\textcolor{white}{0}\\
			\textcolor{white}{0}\\
			\leftarrow (h_1+1)\text{-th row} \\			
			\textcolor{white}{\vdots} \\
			\textcolor{white}{0}\\
			\leftarrow (h_2+1)\text{-th row} \\
			\textcolor{white}{\vdots}\\
			\textcolor{white}{0}\\
			\leftarrow (h_{\ell}+1)\text{-th row} \\
			\textcolor{white}{\vdots}\\
			\textcolor{white}{0}\\
		\end{array}
        %$}
	\end{equation}
\begin{proof}
Consider $(c_1,c_2,\dots,c_n)\in\tgoppa$. 
Since $\gcd(z-\alpha_i,g(z))=1$, $(z-\alpha_i)^{-1}= -g(\alpha_i)^{-1}\left(\frac{g(z)-g(\alpha_i)}{z-\alpha_i}\right)$ under $\bmod\; g(z).$ Now $\frac{g(z)-g(\alpha_i)}{z-\alpha_i}= \sum\limits_{j=1}^{t} g_j \sum\limits_{k=0}^{j-1}\alpha_i^{k} z^{j-1-k}$, and 
	thus $(c_1,c_2,\dots,c_n)\in\tgoppa$ implies 
        % \begin{multline}
        \begin{equation}
        \label{paritycheck derivation}
		\sum_{i=1}^n{ c_i g(\alpha_i)^{-1}\bigg( \sum_{j=1}^{t} g_j \left(\sum_{k=0}^{j-1}\alpha_i^k z^{j-1-k}  \right) + \sum_{j=1}^\ell \eta_j \alpha_i^{t-1+t_j}%\\ 
        \operatorname{U}^{h_j+1}(\bm{z})\cdot \bm{g}\bigg)} \equiv0 \bmod g(z).
	\end{equation}
        % \end{multline}
	Since the polynomial on the left side of \eqref{paritycheck derivation} has degree strictly less than $t$, i.e. degree of $g(z)$, it is indeed the zero polynomial. 
	Thus, the coefficients of $z^i$ are all zero for each $i$. 
	In particular, we obtain $L(H_1+H_2)D(c_1,c_2,\dots,c_n)^T=\bm 0$, where
	\begin{equation*}
        \resizebox{0.5\linewidth}{!}{
		$L = \left( \begin{array}{cccc}
			g_t & 0  & \cdots & 0   \\
			g_{t-1} & g_t & \cdots & 0 \\
			\vdots & \vdots & \vdots & \vdots \\
			g_1 & g_2 & \cdots & g_t  \\
		\end{array} \right),\quad$
	$H_1 = \left( \begin{array}{cccc}
		1 & 1  & \cdots & 1   \\
		\alpha_1 & \alpha_2 & \cdots & \alpha_n \\
		\vdots & \vdots & \vdots & \vdots \\
		\alpha_1^{t-1} & \alpha_2^{t-1}  & \cdots & \alpha_n^{t-1} \\
	\end{array} \right),$}
        \end{equation*}
	\begin{eqnarray}
		H_2 = \left( \begin{array}{cccc}
			\vdots & \vdots & \vdots & \vdots \\
			\eta_1\alpha_1^{t-1+t_1} & \eta_1\alpha_2^{t-1+t_1} & \cdots & \eta_1\alpha_n^{t-1+t_1}\\
			\vdots & \vdots & \vdots & \vdots \\
			
			\eta_2\alpha_1^{t-1+t_2} & \eta_2\alpha_2^{t-1+t_2} & \cdots & \eta_2\alpha_n^{t-1+t_2}\\
			\vdots & \vdots & \vdots & \vdots \\
			
			\eta_\ell\alpha_1^{t-1+t_\ell} & \eta_\ell\alpha_2^{t-1+t_\ell} & \cdots & \eta_\ell\alpha_n^{t-1+t_\ell}\\
			\vdots & \vdots & \vdots & \vdots 
		\end{array} \right), \nonumber \quad
		D = \left( \begin{array}{cccc}
			g(\alpha_1)^{-1} & 0  & \cdots & 0   \\
			0 & g(\alpha_2)^{-1} & \cdots & 0 \\
			\vdots & \vdots & \vdots & \vdots \\
			0 & 0 & \cdots & g(\alpha_n)^{-1}  
		\end{array} \right).
	\end{eqnarray}
    Note that the only non-zero rows of $H_2$ are indexed by $h_i+1$ for $1\leq i\leq \ell$.
	Since $L$ is invertible as $g_t\neq0$, $(H_1+H_2)D$ is the required matrix $H$. This completes the proof.
\end{proof}
	
\end{prop}
Taking into account a particular case where $\ell = 1$, $t_1 = 1$, $h_1 = t - 1$, $\eta_1 = \eta$ and $g_t=1$, it is straightforward to verify that the above definition and the corresponding result coincide with those presented in \cite{Sui2023} and \cite{Sun2025}.

\begin{example}\label{example1}
    Let the base field be $\mathbb{F}_{4}\cong \frac{\mathbb{F}_2[x]}{\langle x^2+x+1\rangle}=\mathbb{F}_2(a)$ and let $\mathbb{F}_{4^2}\cong\frac{\mathbb{F}_4[x]}{\langle x^2+x+a\rangle}=\mathbb{F}_2(a,b)$.
    Consider $g(z)=z^{3} + \left(a b + a + 1\right) z^{2} + \left(\left(a + 1\right) b + a + 1\right) z + a + 1$, a $3$-degree polynomial over $\mathbb{F}_{4^2}$. 
    Consider $\mathcal{L}=\{b + 1, \left(a + 1\right) b + 1, \left(a + 1\right) b, a b + 1, a, a + 1, \left(a + 1\right) b + a + 1, a b, 0, b, \left(a + 1\right) b + a, b + a, b + a + 1, 1\}$.
    Let $\ell=2$ and $\bm{\mathbbm{t}}=(1,2)$, $\bm{\mathbbm{h}} = (1,2)$ and $\bm \eta =(b,ab+1)\in \mathbb{F}_{16}^2$.\\
    Then the matrix $H$, as described in Proposition \ref{prop::parity_mat_mtgoppa}, is given as: \\
    %shown explicitly in Appendix \ref{appendix1}.
    \resizebox{\textwidth}{!}{
    $\left(
    \begin{array}{cccccccccccccc}
    1 & a & 1 & a b & 1 & \left(a + 1\right) b & \left(a + 1\right) b & a b & a & a b + 1 & a b + 1 & a & \left(a + 1\right) b & a b + 1 \\
    \left(a + 1\right) b & a b + a + 1 & a b + 1 & 1 & b + a & b + 1 & a b + 1 & \left(a + 1\right) b + a & 0 & a b + 1 & a + 1 & b + a & \left(a + 1\right) b + 1 & b + a \\
    a b & b + a + 1 & b + a & a b + a & \left(a + 1\right) b + 1 & 1 & a b & b + a + 1 & 0 & 0 & a & b + a & 1 & b + 1
    \end{array}
    \right)$.}
\end{example}

% \begin{prop}
% 	The dimension of multi-twisted Goppa code $\Gamma (\mathcal{L},g,\bm{\bm{t}},\bm{\mathbbm{h}},\bm \eta)$ as a vector space over $\mathbb{F}_q$ is at least $n-mt$, and its Hamming distance is at least $t-\ell+1$.
% 	\begin{proof}
% 		It can be observed that a parity-check matrix of $\Gamma (\mathcal{L},g,\bm{\bm{t}},\bm{\mathbbm{h}},\bm \eta)$ can be obtained by expanding each entry of $H$ as $m$-tuple over $\mathbb{F}_q$ and then considering only the linearly independent rows. The matrix thus obtained will have at most $mt$ rows, which implies $n-\dim \Gamma (\mathcal{L},g,\bm{\bm{t}},\bm{\mathbbm{h}},\bm \eta) \leq mt$.
% 		Additionally, since every $t-\ell$ columns of $H$ are linearly independent (considering them as rows of a non singular Vandermonde matrix), the Hamming distance is at least $t-\ell+1$.
% 	\end{proof}
% \end{prop}

\begin{remark}
    It can be observed that a parity-check matrix of $\tgoppa$ can be obtained by expanding each entry of $H$ as $m$-tuple over $\mathbb{F}_q$ and then considering only the linearly independent rows. The matrix thus obtained will have at most $mt$ rows, which implies $\dim \tgoppa \ge n-mt$.
\end{remark}
Next, we study the sufficient conditions so that $\tgoppa$ has a Hamming distance at least $t+1$.
\begin{theorem}\label{thm::dist_t}
	The multi-twisted Goppa code $\tgoppa$ has minimum distance at least $t+1$ if for each $t$ sized set $T \subset \{1,2,\dots ,n\}$  correspondingly the polynomial $\prod\limits_{i\in T }(x-\alpha_i)= \sum\limits_{i} \sigma_i x^i$ with $\sigma_i =0$ for $i<0$ and $i>t$,	
	the $t_\ell \times t_\ell$ sized matrix
	\begin{equation*}
		\text{diag }\Big( \underset{\underset{t_{\ell}^{th}}{\uparrow} }{\eta_{\ell}^{-1}}, 1,\dots, 1, \underset{\underset{t_{\ell-1}^{th}}{\uparrow} }{\eta_{\ell-1}^{-1}}, 1,\dots, 1, \underset{\underset{t_{1}^{th}}{\uparrow} }{\eta_{1}^{-1}}, 1,\dots, \underset{\underset{1^{st}}{\uparrow} }{1} \Big)
        \cdot A_{T } - 
        \Big( \underset{\underset{t_{\ell}^{th}}{\uparrow} }{{\bm e}_{t_\ell -1-h_\ell}}, \bm 0,\dots, \bm 0, \underset{\underset{t_{\ell-1}^{th}}{\uparrow} }{{\bm e}_{t_\ell -1-h_{\ell-1}}}, \bm 0,\dots,\bm 0, \underset{\underset{t_{1}^{th}}{\uparrow} }{{\bm e}_{t_\ell -1-h_1}},\bm 0,\dots, \underset{\underset{1^{st}}{\uparrow} }{\bm 0} \Big)^\mathsf{T}\cdot B_{T}
	\end{equation*}
	is non-singular, where $\bm e_j\in \mathbb{F}_2^{t_\ell\times 1}$ is the unit vector with $1$ at $j$-th coordinate. 
	Here, $A_{T }$ and $B_{T }$ are lower and upper-triangular $t_\ell\times t_\ell$ matrices, respectively, given by: 
	\begin{equation*}
		A_{T } = \left(\begin{array}{lllll}
			1& 0 &0 & \cdots & 0\\
			\sigma_{t-1}& 1 & 0 &\cdots &0\\
			\sigma_{t-2}& \sigma_{t-1} & 1 &\cdots &0\\
			~~\vdots&~\vdots&~\vdots&~~\ddots&~\vdots\\
			\sigma_{t - t_\ell +1} & \cdots &\cdots & \sigma_{t-1}&1\\
		\end{array}\right) \text{ and }
        B_{T } = 
            \left(\begin{array}{lllll}
			0& \sigma_0 & \sigma_1 & \cdots & \sigma_{t_\ell-2}\\
			0& 0 & \sigma_0 &\cdots &\sigma_{t_\ell-3}\\
			~~\vdots&~\vdots&~\vdots&~~\ddots&~\vdots\\
			0 & \cdots &\cdots & 0&\sigma_0\\
                0 & \cdots &\cdots & 0&0\\
		\end{array}\right).
        \end{equation*}
\end{theorem}
% }

\begin{proof}
	It suffices to analyse the conditions so that every $t$ columns of $H$, described in (\ref{pc-mat-MTG}), are linearly independent over $\mathbb{F}_q$.
	Suppose $T \subset\{1,2,\dots,n\}$ denotes the arbitrary $t$ column indices and consider $H_{T }$ as the corresponding $t\times t$ sized submatrix.
	We require the conditions so that this matrix becomes non-singular, which ensures Hamming distance of the code to be at least $t+1$.
	Using the fact that the $t\times t$ matrix $H_{T }$ is non-singular if and only if all its rows are linearly independent, we
	consider 
	\begin{multline*}
		f_0 + f_1\alpha_{w} + f_2\alpha_{w}^2 +\cdots + f_{h_1}(\alpha_{w}^{h_1} +\eta_1 \alpha_w^{t-1+t_1})+ f_{h_1+1}\alpha_{w}^{h_1+1} +\cdots %\\
        + f_{h_\ell}(\alpha_{w}^{h_\ell} +\eta_\ell \alpha_w^{t-1+t_\ell})+ f_{h_\ell+1}\alpha_{w}^{h_\ell+1} +\cdots + f_{t-1}\alpha_{w}^{t-1}=0
	\end{multline*} 
	for $w\in T $. In other words, for $f(x)=f_0+f_1x+\cdots+f_{t-1}x^{t-1}+f_{h_1}\eta_1x^{t-1+t_1}+\cdots+f_{h_\ell}\eta_\ell x^{t-1+t_\ell}$, we consider $f(\alpha_w)=0$ for $w\in T $. 
	The non-singularity of $H_{T }$ follows if $f\equiv 0$.
	In this direction, let $f(x)=\sigma(x)\gamma(x),$ where $\sigma(x)=\prod_{i\in T } (x-\alpha_i)$ and $\gamma(x)=\sum_{i=0}^{t_\ell-1} \gamma_i x^i$. 
	Thus all the coefficients of $\{x^t,x^{t+1},\dots, x^{t+t_\ell-1}\}\setminus\{x^{t-1+t_1},x^{t-1+t_2},\dots, x^{t-1+t_\ell}\}$ in $f(x)$ are zero, 
	whereas, the remaining  $\ell$ coefficients equals $\eta_if_{h_i}$ for $1\leq i \leq \ell$.
	This generates a system of $t_\ell$ equations in $t_\ell$ variables $\gamma_0,\gamma_1,\dots,\gamma_{t_\ell-1}$ as 
		$\sum_{j=0}^i \sigma_{i-j}\gamma_j=0,\text{ for } i \in \{t,t+1,\dots,t-1+t_\ell\}\setminus 
                                                \{t-1+t_1,\dots,t-1+t_\ell\}$ and
            $\sum_{j=0}^{t-1+t_s} \sigma_{t-1+t_s-j}\gamma_j=\eta_s f_{h_s},\text{ for } s \in \{1,2,\dots,\ell\},
                                                \text{ where }f_{h_s}=\sum_{i=0}^{h_s} \sigma_{h_s-i}\gamma_{i}$.
	This means we arrive at a homogeneous system with variables $(\gamma_{t_\ell-1},\gamma_{t_\ell-2},\dots,\gamma_0)$; and in order to have $f\equiv 0$, we must have $\gamma(x)\equiv 0$, i.e. the homogeneous system must have a non-singular scalar matrix. Now, merely rearranging the terms completes the proof.

For instance, consider a multi-twisted Goppa code $\tgoppa$ with $t=2$, $\ell=2$, $\bm{\mathbbm{t}}=(2,4)$, $\bm{\mathbbm{h}}=(0,1)$ and $\mathcal{L}=\{\alpha_i\}_{i=1}^n$.
Corresponding to some $2$ sized subset $T\subset \{1,2,\dots,n\}$, consider $f(x)=f_0+f_1x+\eta_1f_0x^3+\eta_2f_1x^5$ with $f(\alpha_w)=0$ for $\alpha_w\in T$. 
Here, $f(x)=\sigma(x)\gamma(x)$ for $\sigma(x)=\prod_{i\in T}(x-\alpha_i)$ and $\gamma(x)\in\mathbb{F}_{q^m}[z]$ with $\deg \gamma(x)=3$. 
On comparing coefficients of $x^2$, $x^3$, $x^4$ and $x^5$ both sides, we obtain the following equalities:
    \begin{align*}
        0 &= \sigma_2\gamma_0+ \sigma_1\gamma_1+ \sigma_0\gamma_2\\
        \eta_1f_0 &= \sigma_2\gamma_1+ \sigma_1\gamma_2+ \sigma_0\gamma_3\\
        0 &= \sigma_2\gamma_2+ \sigma_1\gamma_3\\
        \eta_2f_1 &= \sigma_2\gamma_3
    \end{align*}
    Again, substituting $f_0=\gamma_0\sigma_0$ and $f_1=\sigma_1\gamma_0+\sigma_0\gamma_1$, we obtain the system of linear equations as follows:
    \begin{equation*}
        \left( \begin{array}{cccc}
        \eta_2^{-1}\sigma_2 & 0 & -\sigma_0 & -\sigma_1\\
        \sigma_1 & \sigma_2 & 0 & 0\\
        \eta_1^{-1}\sigma_0 & \eta_1^{-1}\sigma_1 & \eta_1^{-1}\sigma_2 & -\sigma_0\\
        0 & \sigma_0 & \sigma_1 & \sigma_2\\
        \end{array}\right)
        \left(\begin{array}{c}
         \gamma_3\\ \gamma_2 \\ \gamma_1 \\ \gamma_0
         \end{array}\right)=
         \left(\begin{array}{c}
         0\\ 0 \\ 0 \\ 0
         \end{array}\right)
    \end{equation*}
    In this, the left side matrix can be seen as diag$(\eta_2^{-1},1,\eta_1^{-1},1)A_T- (\bm e_2, \bm 0,\bm e_3,\bm 0)^\mathsf{T} B_T$, where 
    \begin{equation*}
        A_T=\left( \begin{array}{cccc}
        \sigma_2 & 0 & 0 & 0\\
        \sigma_1 & \sigma_2 & 0 & 0\\
        \sigma_0 & \sigma_1 & \sigma_2 & 0\\
        0 & \sigma_0 & \sigma_1 & \sigma_2\\
        \end{array}\right) \quad \textnormal{ and } \quad
        B_T=\left( \begin{array}{cccc}
        0 & \sigma_0 & \sigma_1 & \sigma_2\\
        0 & 0 & \sigma_0 & \sigma_1\\
        0 & 0 & 0 & \sigma_0\\
        0 & 0 & 0 & 0\\
        \end{array}\right).
    \end{equation*}
    Thus, the code has minimum Hamming distance at least $2$ if and only if the above matrix is non-singular corresponding to each  $2$ sized subset $T\subset \{1,2,\dots,n\}$.
\end{proof}
Referring to Theorem~1 in \cite{BeelenStructural2018}, multi-twisted Goppa codes with a Hamming distance of at least $t + 1$ can be constructed, as stated in the following theorem.

\begin{theorem}\label{thm::mtg_dist}
    Let $\mathcal{L}=\{\alpha_1,\dots,\alpha_n\}\subseteq \mathbb{F}_{s_0}$, and $\mathbb{F}_q\subseteq\mathbb{F}_{s_0}\subsetneq \mathbb{F}_{s_1}\subsetneq \cdots \subsetneq \mathbb{F}_{s_\ell}=\mathbb{F}_{q^m}$ be such that $\bm{\eta}=(\eta_1,\eta_2,\dots,\eta_\ell)$, where $\eta_i\in \mathbb{F}_{s_i}\setminus \mathbb{F}_{s_{i-1}}$ for  $1\leq i\leq \ell$, and $g(z)\in\mathbb{F}_{s_0}[z]$, with $\deg g(z)=t$, and $\bm{\mathbbm{t}}$,  $\bm{\mathbbm{h}}$ be as per Definition \ref{def::MTGC}, then $d(\tgoppa)\geq t+1$.
\end{theorem}
\begin{proof}
    The parity-check matrix $H$ of $\tgoppa$ over $\mathbb{F}_{q^m}$ is described in Proposition \ref{prop::parity_mat_mtgoppa}.
    % , and it can be seen as generator matrix of multi-twisted Reed-Solomon code
    The determinant of an arbitrary $t\times t$ submatrix of $H$ can be expressed as $a(u_\ell \eta_\ell + v_\ell)$, where $a\in \mathbb{F}_{q^m}^*$, $u_\ell,v_\ell\in\mathbb{F}_{s_{\ell-1}}$. Thus, it is zero if and only if both $u_\ell=v_\ell=0$. Proceeding recursively, both $u_\ell$ and $v_{\ell}$ can be expressed as linear combination of $\eta_{\ell-1}$ and $1$ over $\mathbb{F}_{s_{\ell-2}}$, which vanish if and only if corresponding coefficients vanish. In conclusion, the determinant of an arbitrary $t\times t$ sized submatrix of $H$ is zero if and only if a linear combination of $\eta_1$ and $1$ over $\mathbb{F}_{s_0}$ is zero, which is not possible (by assumptions). Therefore, every $t$ columns of $H$ are linearly independent over $\mathbb{F}_{q^m}$, in particular over $\mathbb{F}_q$. This completes the proof.
\end{proof}

\begin{definition}
	[Expurgated Multi-twisted Goppa codes]
	These are subcodes family of multi-twisted Goppa code, denoted by $\widetilde{\Gamma}(\mathcal{L},g, \bm{\mathbbm{t}},\bm{\mathbbm{h}}, \bm{\eta})$, defined as 
	\begin{equation}\resizebox{0.5\linewidth}{!}{$
		\bigg\{(c_1,c_2,\dots,c_n)\in \mathbb{F}_q^n : (c_1,c_2,\dots,c_n)\in \tgoppa \text{ and } \sum_{i=1}^n c_i=0\bigg\}$.}
	\end{equation}
\end{definition}
A parity-check matrix of $\widetilde{\Gamma}(\mathcal{L},g, \bm{\mathbbm{t}},\bm{\mathbbm{h}}, \bm\eta)$, denoted by $\widetilde{H}$, is given by 
\begin{equation}\label{pc-mat-e-goppa}
		\resizebox{0.5\linewidth}{!}{%
        $
                \widetilde{H}:=\left(\begin{array}{cccc}
			g(\alpha_1)^{-1} & g(\alpha_2)^{-1} & \cdots & g(\alpha_n)^{-1}\\
			\alpha_1g(\alpha_1)^{-1} & \alpha_2g(\alpha_2)^{-1} & \cdots & \alpha_ng(\alpha_n)^{-1} \\
			\vdots & \vdots & \cdots & \vdots \\
			(\alpha_1^{h_1}+ \eta_1 \alpha_1^{t-1+t_1})g(\alpha_1)^{-1} & (\alpha_2^{h_1}+ \eta_1 \alpha_2^{t-1+t_1})g(\alpha_2)^{-1} & \cdots & (\alpha_n^{h_1}+ \eta_1 \alpha_n^{t-1+t_1})g(\alpha_n)^{-1} \\
			\alpha_1^{h_1+1}g(\alpha_1)^{-1}&\alpha_2^{h_1 +1}g(\alpha_2)^{-1}&\cdots&\alpha_n^{h_1 +1}g(\alpha_n)^{-1}\\
			\vdots &\vdots &\vdots &\vdots \\
			(\alpha_1^{h_2}+ \eta_2 \alpha_1^{t-1+t_2})g(\alpha_1)^{-1} & (\alpha_2^{h_2}+ \eta_2 \alpha_2^{t-1+t_2})g(\alpha_2)^{-1} & \cdots & (\alpha_n^{h_2}+ \eta_2 \alpha_n^{t-1+t_2})g(\alpha_n)^{-1}\\
			\alpha_1^{h_2+1}g(\alpha_1)^{-1}&\alpha_2^{h_2 +1}g(\alpha_2)^{-1}&\cdots&\alpha_n^{h_2 +1}g(\alpha_n)^{-1}\\
			\vdots &\vdots &\vdots &\vdots \\
			(\alpha_1^{h_{\ell}}+ \eta_{\ell} \alpha_1^{t-1+t_{\ell}})g(\alpha_1)^{-1} & (\alpha_2^{h_{\ell}}+ \eta_{\ell} \alpha_2^{t-1+t_{\ell}})g(\alpha_2)^{-1} & \cdots & (\alpha_n^{h_{\ell}}+ \eta_{\ell} \alpha_n^{t-1+t_{\ell}})g(\alpha_n)^{-1}\\
			\alpha_1^{h_{\ell}+1}g(\alpha_1)^{-1}&\alpha_2^{h_{\ell} +1}g(\alpha_2)^{-1}&\cdots&\alpha_n^{h_{\ell} +1}g(\alpha_n)^{-1}\\
			\vdots &\vdots &\vdots &\vdots \\
			\alpha_1^{t-1}g(\alpha_1)^{-1} & \alpha_2^{t-1}g(\alpha_2)^{-1} & \cdots &  \alpha_n^{t-1}g(\alpha_n)^{-1} \\
                1&1&\cdots&1\\
		\end{array}\right)$}.\end{equation}

\begin{prop}
    If $\sum_{j=1}^\ell g_{h_j} \eta_j \alpha_k^{t-1+t_j} = g_t \alpha_k^t$ for $1 \leq k \leq n$, then 
    $\widetilde{\Gamma}(\mathcal{L}, g, \bm{\mathbbm{t}},\bm{\mathbbm{h}}, \bm\eta) = \tgoppa$.
\end{prop}

\begin{proof}
    The last row of $\widetilde{H}$ is $(1, 1, \dots, 1)$, where for $1\le k\le n$, $k$th entry is 
    $1 = g(\alpha_k)^{-1} g(\alpha_k) = g(\alpha_k)^{-1} \left( \sum_{i=0}^t g_i \alpha_k^i \right)$.
    Performing elementary row operations, we obtain the last row as $(z_1, z_2, \dots, z_n)$, where $     z_k = g(\alpha_k)^{-1} \left( g_t \alpha_k^t - \sum_{j=1}^\ell g_{h_j} \eta_j \alpha_k^{t-1+t_j} \right)$, for each $k$. 
    Using the given hypothesis, the last row becomes a linear combination of the preceding rows, 
    proving the desired result.
\end{proof}

By substituting the specific values $\ell = 1$, $t_1 = 1$, $h_1 = t - 1$, and $g_t = 1$, the above result reduces to Proposition~2.6~(i) in \cite{Sui2023}.

\begin{prop}
    Considering $\ell=1$, $t_1=1$ and $g_t=1$ in Definition of $\widetilde{\Gamma}(\mathcal{L},g, \bm{\mathbbm{t}},\bm{\mathbbm{h}}, \bm\eta)$, if $g_{h_1}\eta_1\neq 1$, then $\widetilde{\Gamma}(\mathcal{L},g, \bm{\mathbbm{t}},\bm{\mathbbm{h}}, \bm\eta)=\widetilde{\Gamma}(\mathcal{L},g)$.
\end{prop}
\begin{proof}
    The last row of $\widetilde{H}$, upon row transformations, reduces to $(z_1,z_2,\dots,z_n)$, where $z_k=g(\alpha_k)^{-1}(1-g_{h_1}\eta_1)\alpha_k^t$ for each $k$. 
    If $1-g_{h_1}\eta_1$ is a non-zero constant, then further reductions yield $\widetilde{H}$ to parity-check matrix of $\widetilde{\Gamma}(\mathcal{L},g)$, which completes the proof.
\end{proof}
By choosing the specific hook position $h_1 = t - 1$, we recover Proposition~2.6~(ii)  from \cite{Sui2023} as a special case of our construction.  
% This structural relationship highlights the significance of the proposed framework: by appropriately expurgating redundant codewords, the resulting Expurgated Multi-Twisted Goppa (EMTG) codes can achieve a higher Hamming distance, potentially exceeding $t + 1$, thereby offering improved error-correcting capability compared to classical constructions.

% Towards this, the following result provides sufficient conditions on an Expurgated single twisted Goppa code so that the distance becomes $t+2$.

% {\color{red}
% \begin{prop}
% 	In addition to the conditions of Theorem \ref{thm::dist_t}, if $\sum_{s=1}^\ell g_{h_s} \eta_s \alpha_i^{t-1+t_s}=\alpha_i^t$ for $1\leq i\leq n$, the Hamming distance of $\widetilde{\Gamma}(\mathcal{L},g, \bm{\bm{t}},\bm{\mathbbm{h}}, \bm\eta)$ is $t$ (i.e. $(t+1)$-th row of $H$ is linearly dependent to first $t$ rows). 
% 	Moreover, in case of $\ell=1$ and $t_1=1$, if $\eta_1 g_{h_1}\neq 1$, then $d\left( \widetilde{\Gamma}(\mathcal{L},g, \bm{\bm{t}},\bm{\mathbbm{h}}, \bm\eta)\right)\geq t+2$.
% \end{prop}}

% \subsection{Dual codes of MTRS codes}
\begin{remark}
    Multi-twisted Goppa codes can be viewed as subfield subcodes of the dual of multi-twisted Reed-Solomon codes.
    Moreover, if we consider the underlying finite field $\mathbb{F}_{q^m}=\mathbb{F}_q$ in the Definition \ref{def::MTGC} of multi-twisted Goppa codes, then their (Euclidean) dual corresponds to multi-twisted Reed Solomon codes.
\end{remark}

\begin{remark}
    In Definition~\ref{def::MTGC}, consider the case $m = 1$ and $\deg g(z) = k$. 
    Let $\bm{v} = (g(\alpha_1)^{-1}, g(\alpha_2)^{-1}, \dots, g(\alpha_n)^{-1})$. 
    Under these settings, the dual codes of $\tgoppa$ correspond to the 
    \emph{generalized multi-twisted  Reed-Solomon (MT-GRS) codes} $\mathcal{C}_k(\bm{\alpha}, \bm{\mathbbm{t}}, \bm{\mathbbm{h}}, \bm{\eta}, \bm{v})$, as defined in (\ref{eqn::MT GRS}).
    This construction establishes a direct algebraic correspondence between MTG codes and the family of MT-GRS codes, thereby providing a unified framework that connects Goppa-type and Reed-Solomon-type code structures through the twisting parameters 
$\bm{\mathbbm{t}}$, $\bm{\mathbbm{h}}$, and $\bm{\eta}$.
\end{remark}

\section{Twisted Goppa codes and their Decoding}\label{Section: Twisted Goppa Codes}
Let $\mathcal{L} = \{\alpha_1, \alpha_2, \dots, \alpha_n\} \subseteq \mathbb{F}_{s_0}^*$, where 
$\mathbb{F}_q \subseteq \mathbb{F}_{s_0} \subsetneq \mathbb{F}_{q^m}$, and let 
$\eta \in \mathbb{F}_{q^m} \setminus \mathbb{F}_{s_0}$. 
Consider a polynomial $g(z)$ of degree $t$ over $\mathbb{F}_{s_0}$ such that 
$g(\alpha_i) \neq 0$ for all $1 \leq i \leq n$. 
The \emph{twisted Goppa code} $\Gamma(\mathcal{L}, g, t_1, h, \eta)$ corresponds to the 
single-twist case $(\ell = 1)$ of the multi-twisted Goppa code family defined in 
Definition~\ref{def::MTGC}. 
According to Theorem~\ref{thm::mtg_dist}, the minimum Hamming distance of 
$\Gamma(\mathcal{L}, g, t_1, h, \eta)$ is at least $t + 1$, implying that the code is capable of 
correcting up to $\lfloor t/2 \rfloor$ errors.

In this section, we propose an explicit decoding algorithm for the twisted Goppa code 
$\Gamma(\mathcal{L}, g, t_1, h, \eta)$ capable of correcting up to $\lfloor t/2 \rfloor$ errors. 
The algorithm provides a constructive realization of the theoretical error-correction bound established in 
Theorem~\ref{thm::mtg_dist} and serves as a foundation for extending decoding methods to the broader class of 
multi-twisted Goppa (MTG) codes.

To begin with, let $H\in\mathbb{F}_{q^m}^{t\times n}$ be a parity-check matrix of the code $\Gamma(\mathcal{L}, g, t_1, h, \eta)$. Then $H$ is given by 
\begin{equation}\label{pc-mat}
    \resizebox{0.5\linewidth}{!}{%
    $\left(\begin{array}{cccc}
			g(\alpha_1)^{-1} & g(\alpha_2)^{-1} & \cdots & g(\alpha_n)^{-1}\\
			\alpha_1g(\alpha_1)^{-1} & \alpha_2g(\alpha_2)^{-1} & \cdots & \alpha_ng(\alpha_n)^{-1} \\
			\vdots & \vdots & \cdots & \vdots \\
			(\alpha_1^{h}+ \eta \alpha_1^{t-1+t_1})g(\alpha_1)^{-1} & (\alpha_2^{h}+ \eta \alpha_2^{t-1+t_1})g(\alpha_2)^{-1} & \cdots & (\alpha_n^{h}+ \eta \alpha_n^{t-1+t_1})g(\alpha_n)^{-1} \\
			\alpha_1^{h+1}g(\alpha_1)^{-1}&\alpha_2^{h +1}g(\alpha_2)^{-1}&\cdots&\alpha_n^{h+1}g(\alpha_n)^{-1}\\
			\vdots &\vdots &\cdots &\vdots \\
			\alpha_1^{t-1}g(\alpha_1)^{-1} & \alpha_2^{t-1}g(\alpha_2)^{-1} & \cdots &  \alpha_n^{t-1}g(\alpha_n)^{-1} \\
    \end{array}\right)$}.
\end{equation}
Let $\bm{y} \in \mathbb{F}_q^n$ denote the received vector, which differs from a transmitted codeword 
$\bm{c}$ by at most $t/2$ errors, i.e., $\bm{y} = \bm{c} + \bm{e}$, where 
$\bm{e} \in \mathbb{F}_q^n$ and ${wt}(\bm{e}) \leq t/2$. 
The corresponding syndrome vector is computed as $\bm{s}^T = H \bm{y}^T$, where
\[
s_l = \sum_{i=1}^{n} \alpha_i^{l} g(\alpha_i)^{-1} e_i, \quad 
0 \leq l \leq t - 1, \; l \neq h, \text{ and }
s_h = \sum_{i=1}^{n} \left(\alpha_i^{h} + \eta \, \alpha_i^{t - 1 + t_1}\right) g(\alpha_i)^{-1} e_i.
\]
Let $J \subseteq \{1, 2, \dots, n\}$ denote the set of indices corresponding to nonzero error components, 
i.e., $e_i \neq 0$ if and only if $i \in J$. 
Accordingly, the nonzero contributions to the syndrome arise only from positions indexed by $J$. 
The syndrome polynomial is then defined as
\begin{equation}\label{dfn:syn-poly}
    s(x) := s_{0+h} + s_{1+h} x + \cdots + s_{t-h-1+h} x^{t - h - 1} 
    + s_{t-h+h} x^{t - h} + \cdots + s_{t-1+h} x^{t - 1},
\end{equation}
where subscript indices are taken modulo $t$.

It is clear that $s(x)=0$ if and only if $\bm{e}=\bm{0}$. % based on the fact that the code considered has parity-check matrix consisting of every $t$ columns as linearly independent.
Upon substitution of coefficients, $$  s(x)= \eta \sum_{i\in J}\alpha_i^{t-1+t_1}g(\alpha_i)^{-1}e_i+ \sum_{l=0}^{t-1}\sum_{i\in J} \left( \alpha_i^{(l+h)\bmod t} g(\alpha_i)^{-1}e_i\right)x^l,$$ which upon rearrangements equates to $$\eta \sum_{i\in J}\alpha_i^{t-1+t_1}g(\alpha_i)^{-1}e_i+\sum_{i\in J}g(\alpha_i)^{-1}e_i\sum_{l=0}^{t-1} \left( \alpha_i^{(l+h)\bmod t} \right)x^l.$$ 
Let $\omega(x):=\sum_{l=0}^{t-1} \left( \alpha_i^{(l+h)\bmod t} \right)x^l$, and
define a linear transformation 
\begin{equation}
    \label{eqn:pi}
    \pi:\mathbb{F}_{q^m}[x]_{<t}\rightarrow\mathbb{F}_{q^m}[x]_{<t} \text{ as } x^i\mapsto x^{(i+h)\bmod t} \text{ for } 0\leq i\leq t-1. 
\end{equation}
%$\pi(f(x))=x^hf(x)\bmod (x^t-1)$
It can be observed that $\pi$ also acts as a permutation.
Thus 
\begin{align*}
    \pi(\omega(x))&=\sum_{l=0}^{t-1} \alpha_i^{(l+h)\bmod t} x^{(l+h)\bmod t} =\sum_{l=0}^{t-1} (\alpha_ix)^{(l+h)\bmod t}=\frac{1-(\alpha_ix)^t}{1-\alpha_ix}.
\end{align*}
Thus $s_\pi(x):=\eta \sum_{i\in J}\alpha_i^{t-1+t_1}g(\alpha_i)^{-1}e_i+ \sum_{i\in J}g(\alpha_i)^{-1}e_i \pi(\omega(x))$. 
Note that $s(x)$ can be computed using $s_\pi(x)$ uniquely and vice versa.
This implies $s_\pi(x)\equiv \eta \sum_{i\in J}\alpha_i^{t-1+t_1}g(\alpha_i)^{-1}e_i+ \sum_{i\in J} \frac{g(\alpha_i)^{-1}e_i}{1-\alpha_ix} \bmod x^t$.
Define the error locator polynomial as 
\begin{equation}
    \label{eqn:sigma_x}
    \sigma(x):=\prod_{i\in J}(x-\alpha_i^{-1}).
\end{equation}
Therefore, we obtain\\
\begin{equation*}
    s_\pi(x)\sigma(x)\equiv \left(\eta \sum_{i\in J}\alpha_i^{t-1+t_1}g(\alpha_i)^{-1}e_i\right)\sigma(x)-%\\
    \sum_{i\in J} \left( \alpha_i^{-1}g(\alpha_i)^{-1}e_i \prod_{j\in J\setminus\{i\}}(x-\alpha_j^{-1}) \right) \bmod x^t.
\end{equation*}
Define the error evaluator polynomial $\tau(x)$ as 
\begin{eqnarray}
&\left(\eta \sum_{i\in J}\alpha_i^{t-1+t_1}g(\alpha_i)^{-1}e_i\right)\sigma(x)-\sum_{i\in J} \left( \alpha_i^{-1}g(\alpha_i)^{-1}e_i \prod_{j\in J\setminus\{i\}}(x-\alpha_j^{-1}) \right). \label{eqn:tau_x}
\end{eqnarray}
Then,
\begin{equation}\label{eqn::9}
    \gcd(\tau(x),\sigma(x))=1 \text{ and }  s_\pi(x)\sigma(x)\equiv \tau(x)\bmod x^t.
\end{equation}

\begin{example}\label{example2}
    Considering $\mathbb{F}_{16}$, $\mathcal{L}$ and $g(z)$ as defined in Example \ref{example1}, let $\mathbb{F}_{16^2}=\mathbb{F}_{16}(c)\cong\frac{\mathbb{F}_{2}(a,b)[x]}{\langle x^2+ax+ab\rangle}$. 
    Let $\eta=c \in\mathbb{F}_{16^2}\setminus \mathbb{F}_{16}$, $t_1=1$ and $h_1=1$, then a parity-check matrix of $\Gamma(\mathcal{L},g,t_1,h,\eta)$ over $\mathbb{F}_{16^2}\cong\mathbb{F}_2(a,b,c)$ is given as follows:\\
    \resizebox{0.75\textwidth}{!}{
    $\left(\begin{array}{rrrrrrr}
        1 & a & 1 & a b & 1 & \left(a + 1\right) b & \left(a + 1\right) b  \\
        \left(\left(a + 1\right) b + 1\right) c + b + 1 & \left(\left(a + 1\right) b\right) c + b + a & \left(\left(a + 1\right) b + a\right) c + \left(a + 1\right) b & c + b + 1 & c + a & \left(\left(a + 1\right) b\right) c + a b & \left(b + a + 1\right) c + a + 1 \\
        b + a + 1 & \left(a + 1\right) b + a + 1 & a b + a + 1 & b + a & a + 1 & b & a b + a
    \end{array}\right. \cdots$}\\
    {\flushright\resizebox{0.75\textwidth}{!}{
    $\left.\begin{array}{rrrrrrr}
         a b & a & a b + 1 & a b + 1 & a & \left(a + 1\right) b & a b + 1 \\
         \left(a b + a\right) c + \left(a + 1\right) b + 1 & 0 & b c + \left(a + 1\right) b + a + 1 & \left(a b + a + 1\right) c + b & \left(\left(a + 1\right) b\right) c + a b + a + 1 & a c + b + 1 & \left(a b + 1\right) c + a b + 1 \\
        \left(a + 1\right) b + a & 0 & 1 & b + 1 & a b + a & \left(a + 1\right) b + 1 & a b + 1
    \end{array}\right)$.}\\}
    \noindent Here $\Gamma(\mathcal{L},g,t_1,h,\eta)$ is at least two dimensional subspace of $\mathbb{F}_4^{14}$.
    Let $\bm y = (1, a, a + 1, 1, 1, a + 1, 1, a, 0, 0, 0, 0, 0, a)$ be received after having one error at $14^{th}$ place and magnitude being $a$ and thus
    $\sigma(x)=(x-1)$.
    Then $H\bm y^T=((a + 1)b + a, ((a + 1)b + a)c + (a + 1)b + a, (a + 1)b + a)$ and then the syndrome polynomial $s(x)=\left(((a + 1)b + a)c + (a + 1)b + a\right) + \left((a + 1)b + a\right)x + \left((a + 1)b + a\right)x^2$.
    Rearrangement of terms of $s(x)$, on applying $\pi$, $s_\pi(x)=s(x)$. Taking $\bmod ~x^3$, $s_\pi(x) = c((a+1)b+a) + \frac{(a+1)b+a}{1-x}\bmod x^3$. Here, $\tau(x)=s_\pi(x)\sigma(x)=c((a+1)b+a)\sigma(x)+((a+1)b+a)=((a+1)b+a)(c(x-1)+1)$. 
    We can now readily see that (\ref{eqn::9}) is being satisfied.
\end{example}

\begin{lemma}\label{lemma:decoding_1tg}
    Suppose that there are at most $t/2$ errors, indexed by $J$, during transmission; then the following holds:
    \begin{enumerate}[$(i)$]
        \item $s_\pi(x)\sigma(x)\equiv \tau(x)\bmod{x^t}$
        \item $\gcd(\sigma(x),\tau(x))=1$
        \item $\deg \tau(x)\leq \deg \sigma(x)=|J|\leq t/2$.
    \end{enumerate}
\end{lemma}
Since $\deg \tau(x) \leq \deg \sigma(x)$, on division we obtain $\tau(x)=
%\left(\eta\sum_{i\in J}  \alpha_i^{t-1+t_1}g(\alpha_i)^{-1}e_i\right)
a
\sigma(x)+w(x)$ with $\deg w(x) <\deg \sigma(x)$ for some $a\in\mathbb{F}_{q^m}$. 
Precisely, $\tau(x)$ is described in (\ref{eqn:tau_x}).
Also for $i\in J$, we have $\tau(\alpha_i^{-1})=-g(\alpha_i)^{-1} e_i\alpha_i^{-1} \prod_{f\in J\setminus \{i\}}(\alpha_i^{-1}-\alpha_f^{-1})=-g(\alpha_i)^{-1}e_i\alpha_i^{-1}\sigma'(\alpha_i^{-1})$, where $\sigma'(x)$ is derivative of $\sigma(x)$. 
This gives 
\begin{equation}\label{eqn:error_coord_val}
	e_i=-\frac{\alpha_i\tau(\alpha_i^{-1})}{g(\alpha_i)^{-1}\sigma'(\alpha_i^{-1})}.
\end{equation}
\begin{prop}\label{prop:3conditions}
    Let $\sigma(x)$ and $\tau(x)$ be some polynomials such that they satisfy the conditions of Lemma \ref{lemma:decoding_1tg}. Then $\sigma(x)$ and $\tau(x)$ are error locator and error evaluator polynomials as described in (\ref{eqn:sigma_x}) and (\ref{eqn:tau_x}) if and only if the following hold:
	\begin{itemize}
		\item $J=\{i: \sigma(\alpha_i^{-1})=0, \text{ for } 1\leq i \leq n\}$ and $|J|=\deg \sigma(x)$;
		\item $e_i=-\frac{\alpha_i\tau(\alpha_i^{-1})}{g(\alpha_i)^{-1}\sigma'(\alpha_i^{-1})}$ for $i\in J$, and $e_i=0$ otherwise;
		\item $\tau(x)=\left(\eta \sum_{i\in J}\alpha_i^{t-1+t_1}g(\alpha_i)^{-1}e_i\right)\sigma(x)+w(x)$ with $\deg w(x) <\deg \sigma(x)$.
	\end{itemize}
\end{prop}
\begin{proof}
    % \textcolor{blue}{A short proof}
    Supposing $\sigma(x)$ and $\tau(x)$ are error locator and error evaluator polynomials, i.e. they are expressed via (\ref{eqn:sigma_x}) and (\ref{eqn:tau_x}) respectively. 
    Clearly, the desired properties hold.
    Conversely, we have $\sigma(x)=\prod_{i\in J}(x-\alpha_i^{-1})$. Also, $\tau(\alpha_i^{-1})=w(\alpha_i^{-1})$ for $i\in J$ with $|J|=\deg \sigma(x)$. Since $\deg w(x)<\deg \sigma(x)$, applying Lagrange's interpolation and using the second condition, we get $w(x)=\sum_{i\in J} \left( \alpha_i^{-1}g(\alpha_i)^{-1}e_i \prod_{j\in J\setminus\{i\}}(x-\alpha_j^{-1}) \right)$. This completes the proof.
\end{proof}

\begin{lemma}
	\label{lem:uniqueness_sigma_tau}
	Consider the encoding using a twisted Goppa code $\Gamma(\mathcal{L},g,{t},{h},{\eta})$ as discussed above. Let $\bm y = \bm c+\bm e$, with $wt(\bm e)\leq \lfloor t/2\rfloor$, and $s(x)$ be the syndrome polynomial of $\bm y$. Then there is a unique pair $(\sigma(x),\tau(x))$ satisfying conditions of Proposition \ref{prop:3conditions} up to a constant multiple.
\end{lemma}
\begin{proof}
    % \textcolor{blue}{A short proof}
    Suppose $(\sigma(x),\tau(x))$ and $(\sigma_2(x),\tau_2(x))$ satisfy Proposition \ref{prop:3conditions}. Both tuples provide error vectors of weight at most $t/2$, which means that there are two codewords at a distance of at most $t$ apart, this violates the distance of the code being at least $t+1$, referring to Theorem \ref{thm::mtg_dist}. 
    Hence, both tuples are associated with the unique error vector. 
    Thus, $(\sigma(x),\tau(x))$ and $(\sigma_2(x),\tau_2(x))$ being distinct is only possible when one is a constant multiple of the other.
\end{proof}
The above result shows that there will be a unique pair $(\sigma(x),\tau(x))$ satisfying the conditions of Proposition \ref{prop:3conditions} considering $\sigma(x)$ monic. Now, we use extended Euclid's algorithm to establish such existence.

\subsection{Decoding Algorithm for Twisted Goppa codes}
Suppose $G(x)$ and $S(x)$ be two polynomials over $\mathbb{F}_{q^m}$ such that $G(x)\ne 0$ and $\deg S(x)<\deg G(x).$ Then the \textit{extended Euclidean algorithm} (EEA) computes the remainders $\tau_i(x),$ the quotients $q_i(x)$ and the auxiliary polynomials $u_i(x)$ and $\sigma_i(x)$ as shown below:
\begin{align*}
    &\tau_{-1}(x)=G(x), &&\tau_0(x)=S(x),\\
    &u_{-1}(x)=1, &&u_0(x)=0,\\
    &\sigma_{-1}(x)=0, &&\sigma_0(x)=1.
\end{align*}
\begin{align*}
    &\tau_{i-2}(x)=\tau_{i-1}(x)q_i(x)+\tau_i(x),\;\deg\tau_i(x)<\deg\tau_{i-1}(x)\\
    &u_i(x)=u_{i-2}(x)-q_i(x)u_{i-1}(x)\\
    &\sigma_i(x)=\sigma_{i-2}(x)-q_i(x)
\sigma_{i-1}(x).
\end{align*}
Observe that $\{\deg\tau_i(x)\}$ is a strictly decreasing sequence of non-negative integers. Let $\kappa:=\max\{i: \deg\tau_i(x)\ne 0\}.$ Then it is known that $\tau_\kappa(x)=\gcd(G(x), S(x))=u_\kappa(x)G(x)+\sigma_\kappa(x)S(x).$
\begin{lemma}\label{Euclidean_lemma}
    Let $ \{q_i(x)\}, \{\tau_i(x)\}, \{u_i(x)\}, \{\sigma_i(x)\}$ be sequences obtained from the extended Euclidean algorithm applied to the polynomials $G(x)$ and $S(x)$ as described above. Then
    \begin{enumerate}[$(i)$]
        \item $\sigma_i(x)\tau_{i-1}(x)-\sigma_{i-1}(x)\tau_i(x) = (-1)^i G(x)$, for all $i\in \{ 0,1, \dots, \kappa+1\}$;
        \item  $u_i(x)G(x)+\sigma_i(x)S(x)=\tau_i(x)$, for all $i\in \{-1, 0, \dots, \kappa+1\}$;
        \item $\deg \sigma_i(x)+\deg \tau_{i-1}(x)=\deg G(x)$, for all $i\in \{ 0,1, \dots, \kappa+1\}$;
        \item $u_i(x)\sigma_{i-1}(x)-u_{i-1}(x)\sigma_i(x)=(-1)^{i+1}$, for all $i\in \{ 0,1, \dots, \kappa+1\}$.
    \end{enumerate}
    \begin{proof}
        The proof follows from Lemma 6.2 of \cite{Roth_book}.
    \end{proof}
\end{lemma}

Following the EEA-based framework of~\cite{Roth_book}, towards previous results for twisted Reed-Solomon codes, it had been shown that valid key-equation solutions correspond to a unique EEA output pair. 
In~\cite{Sui2023}, decoding was established for $\deg \sigma(x)+\deg\tau(x)<\deg g(x)$ (up to $\lfloor (t-1)/2 \rfloor$ errors), 
while~\cite{Sun2025} extended this to the boundary case $\deg \sigma(x)+\deg\tau(x)=\deg g(x)$ by expressing the solution as a linear combination of two consecutive EEA pairs. 
Here, we present an explicit proof of this including the boundary case result in the twisted Goppa setting.

\begin{lemma}\label{lem: uniqueness of nu so that sigma and tau comes from ext Euc algo}
    Let $\sigma(x), \tau(x)\in \mathbb{F}_{q^m}[x]\setminus\{0\}$ be such that $\deg \sigma(x)\leq \lfloor t/2\rfloor$ and $\deg \tau(x) \leq \deg \sigma(x)$. Suppose $\sigma(x)$ and $\tau(x)$ satisfy:
    \begin{enumerate}[$(i)$]
        \item $\sigma(x)S(x)\equiv\tau(x)\bmod G(x)$,
        \item $\gcd(\sigma(x),\tau(x))=1$,
        \item $\deg\sigma(x)+\deg\tau(x)<\deg G(x)$.
    \end{enumerate}
    If $\sigma_i(x)$ and $\tau_{i}(x),$ for $i\in \{-1, 0, \dots, \kappa+1\}$ are the polynomials obtained using extended Euclidean algorithm on polynomials $G(x)$ and $S(x),$ then there is a unique $\nu\in \{-1, 0, \dots, \kappa+1\}$ and $\mu\in\mathbb{F}_{q^m}$ such that 
    $$
    \sigma(x)=\mu\sigma_\nu(x), \, \tau(x)=\mu\tau_\nu(x).
    $$
    % \textcolor{red}{Moreover, if $\deg G(x)=t,\;\deg \tau(x)< \frac{t}{2}$ and $\deg \sigma(x)\le \frac{t}{2},$ then $\nu$ is the least index such that $\deg \tau_\nu(x)<\frac{t}{2}.$}
    \end{lemma}
    \begin{proof}
        Consider the case when $t$ is odd.
        In this case, $\deg \sigma(x)\leq \lfloor\frac{t-1}{2}\rfloor$. 
        Since $\deg \sigma(x)+\deg \tau(x)<\deg G(x)$, $\deg \tau(x)<\deg (G(x)=\tau_{-1}(x))$. 
        As observed $\{\deg \tau_i(x)\}$ is a decreasing sequence, there exists a unique $\nu\in\{0,1,\dots, \kappa+1\}$ such that $\deg \tau_\nu(x)\leq \deg \tau(x)<\deg \tau_{\nu-1}(x)$. 
        Using $(ii)$ of Lemma \ref{Euclidean_lemma}, we have $u_\nu(x)G(x)+\sigma_\nu(x)S(x)=\tau_\nu(x)$.
        Also $\sigma(x)S(x)\equiv\tau(x)\bmod G(x)$ implies $u(x)G(x)+\sigma(x)S(x)=\tau(x)$ for some $u(x)$.
        From these two equalities, it follows that 
        \begin{equation*}
            (u_\nu(x)\sigma(x)-u(x)\sigma_\nu(x) )G(x)= \tau_\nu(x)\sigma(x)-\sigma_\nu(x)\tau(x).
        \end{equation*}
        In particular, $\tau_\nu(x)\sigma(x)\equiv \sigma_\nu(x)\tau(x) \bmod G(x)$. 
        Moreover, since degree of both the products is strictly less than $\deg G(x)$, $\tau_\nu(x)\sigma(x)=\sigma_\nu(x)\tau(x)$. 
        Thereafter $u_\nu(x)\sigma(x)=u(x)\sigma_\nu(x)$. 
        Also, using $(iv)$ of Lemma \ref{Euclidean_lemma}, we have $\gcd(u_\nu(x),\sigma_\nu(x))=1$. 
        So there exists a $\mu(x)\in\mathbb{F}_{q^m}[x]$ such that $\sigma(x)=\mu(x) \sigma_\nu(x)$ and $\tau(x)=\mu(x)\tau_\nu(x)$.
        Now, since $\gcd(\tau(x),\sigma(x))=1$, $\mu(x)$ must be a constant, say $\mu \in\mathbb{F}_{q^m}$. 
        Using $(iii)$ of Lemma \ref{Euclidean_lemma}, $\deg\tau_{\nu-1}(x)=\deg G(x)-\deg \sigma_\nu(x)=\deg G(x)-\deg\sigma(x) \geq t-\lfloor\frac{t-1}{2}\rfloor>\lfloor\frac{t-1}{2}\rfloor$. Again, since $\{\deg \tau_i(x)\}$ is a decreasing sequence, $\deg \tau_\nu(x)=\deg \tau(x)\leq \lfloor\frac{t-1}{2}\rfloor=\lfloor t/2\rfloor$.\\
        The case when $t$ is even, i.e. $\lfloor t/2 \rfloor =t/2$. The result follows on the same lines, completing the proof.
    \end{proof}
% In Lemma \ref{lem: uniqueness of nu so that sigma and tau comes from ext Euc algo}, 
It can be observed that if $t$ is even 
% and $\deg \tau(x)<\deg \sigma(x)=\frac{t}{2},$ then there is a unique index $\nu$ and $\mu\in\mathbb{F}_{q^m}$ such that $\tau(x)=\mu\tau_\nu(x)$ and $\sigma(x)=\mu\sigma_\nu(x).$ 
and $\deg \tau(x)=\deg \sigma(x)=\frac{t}{2}$, then Lemma \ref{lem: uniqueness of nu so that sigma and tau comes from ext Euc algo} cannot be applied. 
Assuming similar conditions on these polynomials, we observe that they are also obtained via extended Euclidean algorithm, described in the following theorem.
\begin{theorem}\label{main proposition}
    Let $t=\deg G(x)$ be even and suppose $\sigma(x), \tau(x)\in \mathbb{F}_{q^m}[x]\setminus\{0\}$ satisfy
    \begin{enumerate}[$(i)$]
        \item $\sigma(x)S(x)\equiv\tau(x)\bmod G(x)$,
        \item $\gcd(\sigma(x),\tau(x))=1$,
        \item $\deg\tau(x)\le \deg\sigma(x)={t}/{2}$.
    \end{enumerate}
    If $\sigma_i(x)$ and $\tau_{i}(x),$ for $i\in \{-1, 0, \dots, \kappa+1\}$ are the polynomials obtained using extended Euclidean algorithm on  polynomials $G(x)$ and $S(x),$ then there exist $\mu_1\in\mathbb{F}_{q^m}$ and $\mu_2\in\mathbb{F}_{q^m}^*$ such that
    \begin{align*}
        \sigma(x)&=\mu_1\sigma_{\nu-1}(x)+\mu_2\sigma_{\nu}(x),\\
        \tau(x)&=\mu_1\tau_{\nu-1}(x)+\mu_2\tau_{\nu}(x),
    \end{align*}
    where $\nu$ is the least index such that $\deg \tau_\nu(x)<{t}/{2}.$ 
    Moreover, if $\deg \tau(x)<\deg \sigma(x)={t}/{2}$ in $(iii)$, then $\mu_1=0.$
\end{theorem}
\begin{proof}
    In the initialization of Euclid's algorithm, $\tau_{-1}(x)=G(x)$ and in the termination, $\tau_{\kappa+1}(x)=0$.
    Consider $\nu$ as $$\min\left\{i\in \{0, 1, \dots, \kappa+1\} : \deg \tau_i(x)<\frac{t}{2}\right\}.$$
    From the Euclidean algorithm, we have $S(x)\sigma_i(x)+ u_i(x)G(x)=\tau_i(x)$ for all $i\in\{-1,0,\dots,\kappa+1\}$. In particular, we have 
    \begin{equation}
        \label{eqn:thm3_1}
        S(x)\sigma_\nu(x)\equiv \tau_\nu(x) \bmod G(x).
    \end{equation}
    Since $\gcd(\sigma(x),\tau(x))=1$, there exist $a(x)$ and $b(x)$ s.t. $a(x)\sigma(x)+b(x)\tau(x)=1$, also $\tau(x)=\sigma(x)S(x)+c(x)G(x)$ for some $c(x)\in\mathbb{F}_{q^m}[x]$, we obtain $(a(x)+b(x)S(x))\sigma(x)+b(x)c(x)G(x)=1$, consequently $\gcd(\sigma(x),G(x))=1$.
    Therefore $S(x)\equiv \frac{\tau(x)}{\sigma(x)} \bmod G(x)$. On using this with $(\ref{eqn:thm3_1})$, we obtain
    \begin{equation}
        \label{eqn:thm3_2}
        \tau(x)\sigma_\nu(x)\equiv \tau_\nu(x)\sigma(x)\bmod G(x).
    \end{equation}
    Now we aim at showing $\deg \sigma_\nu(x)=t/2$.
    First, assume $\deg \sigma_\nu(x)<t/2$. Since $\deg \tau(x) \leq \deg \sigma(x)=t/2$, $\deg (\tau(x)\sigma_\nu(x))<t$. 
    As $\deg \tau_\nu(x)<t/2$ and $\deg \sigma(x)=t/2$, we have $\deg(\tau_\nu(x)\sigma(x))<t$. 
    Furthermore $\deg G(x)=t$ implies $\tau(x)\sigma_\nu(x)=\tau_\nu(x)\sigma(x)$.
    Since $\gcd(\sigma(x),\tau(x))=1$, $\sigma(x)|\sigma_\nu(x)$, which implies $\deg \sigma_\nu(x)\geq t/2$, a contradiction to our assumption. 
    Secondly, assume $\deg \sigma_\nu(x)>t/2$. 
    Using $(iii)$ of Lemma \ref{Euclidean_lemma}, i.e. $\deg \sigma_\nu(x)+\deg \tau_{\nu-1}(x)=\deg G(x)=t$, we have $\deg\tau_{\nu-1}(x)<t/2$ a contradiction as $\nu$ is considered smallest and $\{\deg \tau_i(x)\}$ is a strictly decreasing sequence of non negative integers.
    This shows that  $\deg \sigma_\nu(x)=t/2$. \\
    Now, if $\deg \tau(x)< \deg \sigma(x)=t/2$, using (\ref{eqn:thm3_2}), similar to Lemma \ref{lem: uniqueness of nu so that sigma and tau comes from ext Euc algo}, we obtain $\tau(x)\sigma_\nu(x)=\tau_\nu(x)\sigma(x)$. 
    Since $\gcd(\sigma(x),\tau(x))=1$, $\sigma(x)|\sigma_\nu(x)$ and $\deg \sigma(x)=\deg \sigma_\nu(x)$,
    it implies that $\sigma(x)=\mu_2\sigma_\nu(x)$ and $\tau(x)=\mu_2\tau_\nu(x)$ for some $\mu_2\in\mathbb{F}_{q^m}^*$.\\
    Since $\deg \sigma_\nu(x)+\deg\tau_{\nu-1}(x)=t$, using $(iii)$ of Lemma \ref{Euclidean_lemma}, $\deg \tau_{\nu-1}(x)=t/2$. 
    Also since $\{\deg \tau_i(x)\}$ is a strictly decreasing sequence, we obtain $\deg \sigma_{\nu-1}(x)+\deg \tau_{\nu-1}(x)<t$, and in particular $\deg \sigma_{\nu-1}(x)<t/2$.\\
    Since $\sigma_{\nu-1}(x)S(x)\equiv \tau_{\nu-1}(x)\bmod G(x)$, $\sigma_{\nu-1}(x)\frac{\tau(x)}{\sigma(x)}\equiv \tau_{\nu-1}(x)\bmod G(x)$, equivalently,
    \begin{equation}
        \label{eqn:thm3_3}
        \sigma_{\nu-1}(x)\tau(x)\equiv \tau_{\nu-1}(x)\sigma(x)\bmod G(x).
    \end{equation}
    Furthermore, if $\deg \tau(x)=\deg \sigma(x)=t/2$, then since $\deg \tau_{\nu-1}(x)=\deg \sigma_\nu(x)=t/2$,  $\deg \sigma_{\nu-1}(x)<t/2$ and $\deg \tau_{\nu-1}(x)<t/2$, we obtain
    \begin{align*}
        G(x)=&\mu_1(\tau_{\nu-1}(x)\sigma(x)-\sigma_{\nu-1}(x)\tau(x))\\
        G(x)=&-\mu_2(\tau_{\nu}(x)\sigma(x)-\sigma_{\nu}(x)\tau(x)),
    \end{align*}
    for some $\mu_1,\mu_2\in \mathbb{F}_{q^m}^*$.
    Finally since $\gcd(\sigma(x),\tau(x))=1$, on equating the above equalities, we obtain $\sigma(x)=\mu_1\sigma_{\nu-1}(x)+\mu_2 \sigma_\mu(x)$ and $\tau(x)=\mu_1 \tau_{\nu-1}(x)+\mu_2\tau_\nu(x)$.
\end{proof}

Consider the twisted Goppa code $\mathcal{C}:=\Gamma(\mathcal{L},g,{t},{h},{\eta}).$ Suppose a codeword of $\mathcal{C}$ is transmitted and $\bm{r}$ is a received word with at most $t/2$-erroneous components, that is, $d(\bm{r}, \mathcal{C})\le \frac{t}{2}.$ Let $s(x)$ be the syndrome polynomial of the received vector $\bm{r}.$ Based on Proposition \ref{prop:3conditions}, we determine the error-locator and the error-evaluator polynomial of the received word $\bm{r}$ using Lemma \ref{lem: uniqueness of nu so that sigma and tau comes from ext Euc algo} and Theorem \ref{main proposition}. %hence decode the received word $\bm{r}$.
\begin{theorem}[Decoding]\label{thm:decoding}
    Let $\mathcal{C}$ be the twisted Goppa code $\Gamma(\mathcal{L},g,{t},{h},{\eta})$. Let $\bm{r}\in\mathbb{F}_{q}^n$ be such that $d(\bm{r}, \mathcal{C})\le \frac{t}{2}$ and let $s(x)\neq 0$ be its syndrome, as defined in (\ref{dfn:syn-poly}). 
    Suppose $\sigma_i(x)$ and $\tau_i(x)$ are polynomials from extended Euclidean algorithm applied to $s_{\pi}(x)$ and $x^t$ for $i\in \{-1, 0, \dots, \kappa+1\}$ and let $\nu:=\min\left\{i\in \{0, 1, \dots, \kappa+1\} : \deg \tau_i(x)<\frac{t}{2}\right\}.$
    \begin{enumerate}[$(i)$]
        \item If $\deg \sigma_\nu(x)<\frac{t}{2},$ then $(\mu\sigma_\nu(x), \mu\tau_\nu(x))$ satisfy the conditions of Lemma \ref{lemma:decoding_1tg} and in fact, are the error-locator and error evaluator polynomials of the received word $\bm{r}.$ Here, $\mu\in\mathbb{F}_{q^m}$ is chosen such that $\mu\sigma_\nu(x)$ becomes monic.
        \item If $t$ is even and $\deg \sigma_\nu(x)=\frac{t}{2},$ then there exist $\mu_1,\mu_2\in\mathbb{F}_{q^m}$, s.t. $(\mu_1\sigma_{\nu-1}(x)+\mu_2\sigma_{\nu}(x), \mu_1\tau_{\nu-1}(x)+\mu_2\tau_\nu(x))$ satisfy the conditions of Lemma \ref{lemma:decoding_1tg} and in fact, are the error-locator and error evaluator polynomials of the received word $\bm{r}.$ The constants $\mu_1$ and $\mu_2$ are chosen such that the corresponding tuple satisfies the conditions of Proposition \ref{prop:3conditions}.
    \end{enumerate}
\end{theorem}
\begin{proof}
    Let $\sigma(x)$ and $\tau(x)$ be the error-locator and the error-evaluator polynomial of the received word $\bm{r}.$ 
    Then $(\sigma(x),\tau(x))$ satisfy the conditions of Lemma \ref{lemma:decoding_1tg} taking $G(x)=x^t$ and $S(x)=s_\pi(x)$. 
    Consequently, $\gcd(\sigma(x), x^t)=1$ and
    \begin{equation}\label{cong eqn 1}
        s_\pi(x)\equiv \tau(x)\sigma(x)^{-1}\bmod x^t.
    \end{equation}
    From Lemma \ref{Euclidean_lemma}(ii), 
    \begin{equation}\label{cong eqn 2}
        s_\pi(x)\sigma_\nu(x)\equiv \tau_\nu(x)\bmod x^t.
    \end{equation}
    Equations \eqref{cong eqn 1} and \eqref{cong eqn 2} imply
    \begin{equation}
        \tau(x)\sigma_\nu(x)\equiv \sigma(x)\tau_\nu(x) \bmod x^t,
    \end{equation} and if $\deg \sigma_\nu(x)<\frac{t}{2}$, we obtain
    \begin{equation}
        \tau(x)\sigma_\nu(x)=\sigma(x)\tau_\nu(x).
    \end{equation}
    Since $\gcd(\tau(x), \sigma(x))=1,\;\sigma(x)|\sigma_\nu(x)$ and $\tau(x)|\tau_\nu(x)$. Consequently, $\deg \sigma(x)\le \deg \sigma_\nu(x)<\frac{t}{2}.$ 
    Applying Lemma \ref{lem: uniqueness of nu so that sigma and tau comes from ext Euc algo} to $\sigma(x)$ and $\tau(x)$, we obtain, $\sigma(x)=\mu\sigma_\nu(x)$ and $\tau(x)=\mu\tau_\nu(x)$ for suitable choice of $\mu\in\mathbb{F}_{q^m}$.
    % Choosing $\mu$ suitably, $\sigma(x)=\mu\sigma_\nu(x)$ and $\tau(x)=\mu\tau_\nu(x).$ 
    This proves (i).\\
    %this means degree of $\sigma(x)$ is already less than $t/2$, which we can say only if we know the number of errors initially.}\\
    If $\deg \sigma_\nu(x)=\frac{t}{2},$ then by Lemma \ref{Euclidean_lemma}(iii), $\deg \tau_{\nu-1}(x)=\deg G(x)-\deg \sigma_\nu(x)=\frac{t}{2}.$ 
    Since $\deg \tau_{\nu-2}(x)>\deg \tau_{\nu-1}(x),$ again by Lemma \ref{Euclidean_lemma}(iii), $\deg \sigma_{\nu-1}(x)<\frac{t}{2}.$ 
    Suppose that $\deg \sigma(x)<\frac{t}{2},$ then by Lemma \ref{lem: uniqueness of nu so that sigma and tau comes from ext Euc algo}, $\sigma(x)=\mu \sigma_{\nu}(x).$ Consequently, $\deg \sigma_{\nu}(x)<\frac{t}{2},$ a contradiction to the hypothesis. 
    Thus, $\deg \sigma(x)=\frac{t}{2}.$ 
    By Theorem \ref{main proposition}, $\sigma(x)=\mu_1\sigma_{\nu-1}(x)+\mu_2\sigma_{\nu}(x)$ and $\tau(x)=\mu_1\tau_{\nu-1}(x)+\mu_2\tau_\nu(x)$, for suitable $\mu_1,\mu_2\in\mathbb{F}_{q^m}$ s.t. the corresponding tuple satisfies the conditions of Proposition \ref{prop:3conditions}.
    %Using (??), we have $\sigma_{\nu-1}(x)s(x)\equiv \tau_{\nu-1}(x)\bmod g(x)$, and further (\ref{cong eqn 1}) implies 
    %\begin{equation*}
        %\tau_{\nu-1}(x)\sigma(x)\equiv \tau(x)\sigma_{\nu-1}(x)\bmod g(x)
    %\end{equation*}
    %\textcolor{blue}{Use Theorem 3. Why are we not using it? }
    %{\color{red}To use Theorem 3, $\deg \sigma(x)$ must be equal to $\frac{t}{2}.$ Why can't we assume $t/2$ errors occurred during transmission? This will give degree of $\sigma(x)$ equal to $t/2$?}
    %{\color{red} Stuck here!!!!!}\\
    %{\color{blue} Suppose during transmission $t'\leq t/2$ errors occurred, which gives $\deg \sigma(x)=t'$ and $\deg \tau(x)\leq \deg \sigma(x)=t'\leq t/2$. \\
    %Case 1. If $t'<t/2$, then we have two sub-cases:\\
    %Sub case 1: $\deg \sigma_\nu(x)<t/2$, done!\\
    %Sub case 2: $\deg \sigma_\nu(x)\geq t/2$, this is not possible! (How?)\\
    %{\color{red} These subcases will not arise since $\deg \sigma_\nu(x)$ is exactly $\frac{t}{2}.$}
    %Case 2: If $t'=t/2$, then again we have two sub-cases:\\
    %Sub case 1: $\deg \sigma_\nu(x)\leq t/2$, both can be handled by Theorem 3.\\
    %Sub case 2: $\deg \sigma_\nu(x)> t/2$, this is not possible! (How?)\\}

    If $\deg \sigma_\nu(x)>\frac{t}{2},$ then by Lemma \ref{Euclidean_lemma}(ii), $\deg \tau_{\nu-1}(x)<\frac{t}{2},$ a contradiction, since $\nu$ is  the  least index such that $\deg \tau_{\nu}(x)<\frac{t}{2}.$
\end{proof}

\begin{remark}
    It is to be noted here that, taking $\eta=0$, the decoding algorithm described above is also valid for classical Goppa codes.
\end{remark}

\begin{algorithm}[H]
\caption{Decoding a Twisted Goppa Code with up to $\lfloor t/2 \rfloor$ Errors}
\label{alg:twisted-goppa-decoder}
\begin{algorithmic}[1]
  \Require Parity-check matrix $H \in \mathbb{F}_{q^m}^{t \times n}$ of $\mathcal{C}=\Gamma(\mathcal{L},g,t,h,\eta)$ (as in~(\ref{pc-mat})); received word $\bm{r} \in \mathbb{F}_q^n$
  \Ensure Error vector $\bm{e}$ and codeword $\bm{c}=\bm{r}-\bm{e}$

  \State \textit{Syndrome:} Compute $\bm{s}^T = H \bm{r}^T$ and define
  \[
    s(x) := s_{0+h} + s_{1+h} x + \cdots + s_{t-h-1+h} x^{t - h - 1} 
    + s_{t-h+h} x^{t - h} + \cdots + s_{t-1+h} x^{t - 1},
  \] where subscript indices are taken modulo $t$.
  \If{$s(x)\equiv 0$}
    \State \Return $\bm{e}=\bm{0}$, $\bm{c}=\bm{r}$
  \EndIf

  \State \textit{Twist transform:} Compute $s_\pi(x)\gets \pi(s(x))$ as defined in~(\ref{eqn:pi}).

  \State \textit{Initialize EEA:} Set 
  \[
    \tau_{-1}(x)=x^{t},\;\; \tau_{0}(x)=s_\pi(x),\;\; 
    \sigma_{-1}(x)=0,\;\; \sigma_{0}(x)=1,\;\;
    u_{-1}(x)=1,\;\; u_{0}(x)=0.
  \]

  \State\label{l7} \textit{EEA step:} Apply the extended Euclidean algorithm on $(x^{t}, s_\pi(x))$ to obtain triples $(u_i(x), \tau_i(x), \sigma_i(x))$ for $i=1,2,\dots,\kappa$, where
  \[
    \kappa \coloneqq \max\{i \mid \deg \tau_i(x) \ge 0\}.
  \]

  \State\label{l8} \textit{Index selection:} Let
  \[
    \nu \coloneqq \min\{i \in \{0,1,\dots,\kappa\} \mid \deg \tau_i(x) < t/2\}.
  \]

  \State \textit{Case A:} If $\deg \sigma_\nu(x) < t/2$ then
  \State Choose $\mu \in \mathbb{F}_{q^m}^*$ so that $\mu\sigma_\nu(x)$ is monic, and set
  \[
    \sigma(x)\gets \mu\sigma_\nu(x), \qquad \tau(x)\gets \mu\tau_\nu(x).
  \]
  \State $J \gets \{\, i \mid \sigma(\alpha_i^{-1})=0,\; 1\le i\le n \,\}$ and ensure $\gcd(\sigma,\sigma')=1$.
  \For{$i\in J$}
        \State $e_i \gets -\,\alpha_i\,g(\alpha_i)\,\dfrac{\tau(\alpha_i^{-1})}{\sigma'(\alpha_i^{-1})}$;\quad 
               \textbf{else } $e_i\gets 0$
      \EndFor
  \State \Return $\bm{e}=(e_1,\dots,e_n)$ and $\bm{c}=\bm{r}-\bm{e}$.

  \State \textit{Case B:} If $\deg \sigma_\nu(x)=t/2$ then
  \State Choose $\mu_2\in\mathbb{F}_{q^m}^*$ such that $\mu_2\sigma_\nu(x)$ is monic.
  \For{$\mu_1 \in \mathbb{F}_{q^m}$}
    \State\label{l17} $\sigma(x)\gets \mu_1\sigma_{\nu-1}(x)+\mu_2\sigma_\nu(x)$;\quad
           $\tau(x)\gets \mu_1\tau_{\nu-1}(x)+\mu_2\tau_\nu(x)$
    \State $J\gets \{\, i \mid \sigma(\alpha_i^{-1})=0,\; 1\le i\le n \,\}$
    \If{$\deg \sigma(x)=|J|$ \textbf{ and } $\gcd(\sigma,\sigma')=1$}
      \For{$i\in J$}
        \State $e_i \gets -\,\alpha_i\,g(\alpha_i)\,\dfrac{\tau(\alpha_i^{-1})}{\sigma'(\alpha_i^{-1})}$;\quad 
               \textbf{else } $e_i\gets 0$
      \EndFor
      \State\label{l25} $a \gets \eta\sum_{i\in J} g(\alpha_i)^{-1} e_i \alpha_i^{t-1+t_1}$;\quad 
             $w(x)\gets \tau(x)-a\sigma(x)$
      \If{$\deg w(x)<\deg \sigma(x)$}
        \State \Return $\bm{e}=(e_1,\dots,e_n)$,\; $\bm{c}=\bm{r}-\bm{e}$
      \EndIf
    \EndIf
  \EndFor
\end{algorithmic}
\end{algorithm}

\begin{example}\label{example3} 
    Consider the finite field $\mathbb{F}_{2^5}=\mathbb{F}_2(a)=\mathbb{F}_2[x]/\langle x^5+x^2+1\rangle$. 
    Let $t=3$ and $g(z)=z^{3} + a z^{2} + a^6 z + a^{8}$. 
    Let $\mathcal{L}=\{\alpha_i\}_{i=1}^{20}=\{a^{23},\; a^{3},\; a^{25},\; a^{6},\; a^{21},\; a^{4},\;    a^{14},\; a^{22},\; a^{20},\; a^{8},\;    a^{13},\; a^{16},\; a^{1},\; a^{9},\;    a^{11},\; a^{0},\; a^{28},\; a^{26},\; a^{2},\; a^{12} \}$, i.e. $n=20$.
    Considering $\ell=1$, with $t_1=1$, $h=1$ and $\eta\in\{\mathbb{F}_{(2^5)^2}=\mathbb{F}_{2^5}(c)\}\setminus \mathbb{F}_{2^5}$ be given as $\eta=a^{3} c + a^{17}$. 
    The parity-check matrix $H\in \mathbb{F}_2(a,c)^{3\times 20}$ is described as:\\
    \resizebox{0.9\textwidth}{!}{$
    \left(\begin{array}{rrrrrrrrrr}
        a^{4} & a^{6} & a^{18} & a^{18} & a^{25} & a^{28} & a^{8} & a^{3} & a^{30} & a^{17} \\

        a^{14}c + a^{14} & a^{18}c + a^{21} & a^{3}c + a^{14} & a^{8}c + a^{27} & a^{29}c + a^{10} & a^{12}c + a^{22} & a^{22}c + a^{4} & a^{10}c + a^{11} & c + a^{16} & a^{13}c + a^{30}  \\
        
        a^{19} & a^{12} & a^{6} & a^{30} & a^{5} & a^{5} & a^{5} & a^{16} & a^{8} & a^{2} 
    \end{array}\right.\cdots$}\\
    {\flushright\resizebox{0.9\textwidth}{!}{$
    \left.\begin{array}{rrrrrrrrrr}
        a^{23} & a^{8} & a^{6} & a^{16} & a^{9} & a^{29} & a^{20} & a^{0} & a^{8} & a^{1} \\
                
        a^{3}c + a^{28} & a^{28}c + a^{25} & a^{12}c + a^{18} & a^{15}c + a^{4} & a^{14}c + a^{9} & a^{1}c + a^{28} & a^{14}c + a^{5} & a^{19}c + a^{17} & a^{17}c + a^{4} & a^{9}c + a^{17}\\
        
        a^{18} & a^{9} & a^{8} & a^{3} & a^{0} & a^{29} & a^{14} & a^{21} & a^{12} & a^{25}
    \end{array}\right)$}\\}
    In particular, a generator matrix of $\Gamma(\mathcal{L},g,t_1,h,\eta)$ is given by 
    \begin{equation*}
        G=\left(\begin{array}{rrrrrrrrrrrrrrrrrrrr}
            1 & 0 & 0 & 1 & 0 & 0 & 0 & 1 & 0 & 1 & 1 & 0 & 1 & 0 & 0 & 0 & 1 & 1 & 0 & 0 \\
            0 & 0 & 1 & 1 & 1 & 0 & 0 & 0 & 1 & 0 & 1 & 0 & 1 & 1 & 1 & 0 & 0 & 0 & 1 & 1 \\
            0 & 0 & 0 & 0 & 0 & 1 & 0 & 1 & 0 & 1 & 1 & 1 & 1 & 1 & 1 & 1 & 0 & 0 & 1 & 0 \\
            0 & 0 & 0 & 0 & 0 & 0 & 1 & 0 & 1 & 1 & 0 & 1 & 1 & 0 & 0 & 1 & 1 & 1 & 1 & 1
        \end{array}\right).
    \end{equation*}
    Consider the received vector $y=(0, 0, 0, 1, 0, 0, 0, 1, 0, 1, 1, 0, 1, 0, 0, 0, 1, 1, 0, 0)$ having one error at some place.
    Since $h=1$, here $s(x)= (a^{14}c+a^{14}) + a^{19}x + a^4x^2 =  \eta a^{11} + a^{27}+ a^{19}x + a^4x^2 $,  and $ s_\pi(x)= \eta a^{11} + a^4 + a^{27}x + a^{19}x^2$.
    Applying Theorem $\ref{thm:decoding}$, i.e. in particular on applying extended Euclidean algorithm on $s_\pi(x)$ and $x^3$ till the remainder has degree less than $ 3/2$, we obtain $q_1(x)=a^{12}x+a^{20}$ and $\tau_1(x)=(a^{19}c+a^{15})x+a^{27}c+a^{10}$.
    Referring to the notations, $\sigma(x)=\mu q_1(x)$, where $\mu$ is chosen s.t. $\sigma(x)$ becomes monic, i.e. $\mu=a^{19}$. 
    Therefore, we obtain $\sigma(x)=x-a^8$ and $\tau(x)=(a^7c+a^3)x+a^{15}c+a^{29}$.
    The root of $\sigma(x)$ is $a^8$, i.e. $\alpha_1^{-1}$ so the error position is $1^{st}$ coordinate. 
    The error magnitude, as defined in (\ref{eqn:error_coord_val}), is calculated as $e_1=1$. 
    So the decoded codeword is given as $(1, 0, 0, 1, 0, 0, 0, 1, 0, 1, 1, 0, 1, 0, 1, 0, 0, 1, 1, 1)$.
\end{example}

\subsection{Complexity Analysis of Decoding}
The computational complexity of Algorithm~\ref{alg:twisted-goppa-decoder} is analyzed in terms of the number of elementary operations over the finite field $\mathbb{F}_{q^m}$, namely additions, multiplications, and inversions, following the standard cost model used in algebraic coding theory~\cite{Patterson1975,Berlekamp1984,Blahut2003}. 
The total running time is dominated by polynomial arithmetic (multiplications, divisions, and gcd computations) whose cost can be expressed as $\mathcal{O}(M(t))$, where $M(t)$ denotes the cost of multiplying two polynomials of degree~$t$ over $\mathbb{F}_{q^m}$~\cite{vonzurGathen1999}. 
Referring \cite{vonzurGathen1999, Shoup1995}, $M(t)=\mathcal{O}(t^2)$, while fast algorithms yield $M(t)=\tilde{\mathcal{O}}(t)$~\cite{Nielsen2017}. 
Throughout this section, we consider the classical arithmetic assumptions, so that each polynomial multiplication of degree at most~$t$ costs $\mathcal{O}(t^2)$ field operations.
The decoding process comprises six dominant computational stages: syndrome formation, twist transformation, extended Euclidean iterations, index selection, error-locator polynomial evaluation, and error-magnitude computation.
For each step, we explicitly count the number of $\mathbb{F}_{q^m}$ operations in terms of the code parameters $(n,t)$, and summarize the total decoding complexity for both Case~A and Case~B branches.
The resulting complexity estimates are consistent with those reported in classical analyses of Patterson decoding and its generalizations~\cite{Patterson1975,Berlekamp1984,vonzurGathen1999,Nielsen2017}.

\begin{theorem}[Time complexity in the Case~A, i.e. when there are less than $t/2$ errors]
\label{thm:complexity-caseA}
    Let $\mathcal{C}=\Gamma(\mathcal{L},g,t,h,\eta)$ be a twisted Goppa code over $\mathbb{F}_{q}$,
    with parity-check matrix $H\in\mathbb{F}_{q^m}^{t\times n}$ and decoding algorithm
    given in Algorithm~\ref{alg:twisted-goppa-decoder}.
    Assume the algorithm terminates in Case~A (i.e., $\deg\sigma_\nu(x)<t/2$).
    Then,  the total decoding cost is $\mathcal{O}(t^2 + t n)$    operations over $\mathbb{F}_{q^m}$.
\end{theorem}

\begin{proof}
We count the dominant $\mathbb{F}_{q^m}$-operations step by step.

\begin{itemize}
  \item[{(1)}] \textit{Syndrome computation.} 
  Computing $\bm{s}^T = H\bm{r}^T$ requires $t n$ multiplications 
  and $t(n-1)$ additions in $\mathbb{F}_{q^m}$, 
  i.e.\ $\mathcal{O}(tn)$ total operations. 
  Forming $s(x)$ adds another $\mathcal{O}(t)$ operations.

  \item[{(2)}] \textit{Twist transform.} 
  The map $\pi:\mathbb{F}_{q^m}[x]\to\mathbb{F}_{q^m}[x]$ acts linearly 
  (permutes and scales coefficients), costing $\mathcal{O}(t)$ operations.

  \item[{(3)}] \textit{Extended Euclidean Algorithm.}
    Let $G(x)=x^{t}$ and $S(x)=s_{\pi}(x)\in\mathbb{F}_{q^{m}}[x]$.
    Applying the EEA to $(G,S)$, the degree of the     remainder decreases by \emph{at least} one at each iteration. 
    The procedure can terminate (ref. Lines \ref{l7},\ref{l8}) once $\deg \tau_i(x)<t/2$, i.e., after at most     $\left\lfloor\tfrac{t}{2}\right\rfloor$ iterations. 
    If only remainders are    maintained (without B\'ezout coefficients), a worst-case iteration (when the    degree drops by exactly one) performs $\mathcal{O}(1)$ field inversions and about    $2(d+1)$ multiplications/additions, where $d$ is the current divisor degree. 
    Summing from $d=t-1$ down to $t-\left\lfloor\tfrac{t}{2}\right\rfloor$ yields     $\mathcal{O}\!\bigl(\left\lfloor\tfrac{t}{2}\right\rfloor\bigr)$ inversions and    
    \[
    \left\lfloor\tfrac{t}{2}\right\rfloor\!\bigl(2t-\left\lfloor\tfrac{t}{2}\right\rfloor+1\bigr)
    =
    \begin{cases}
    m(3m+1), & t=2m,\\
    3m(m+1), & t=2m+1,
    \end{cases}
    \]
    multiplications/additions. When the B\'ezout coefficient polynomials $(u_i(x),\sigma_i(x))$ are     also updated, the operation counts increase by an additional quadratic term; 
    a    conservative bound is $\mathcal{O}(t)$ inversions and $\mathcal{O}(t^{2})$     multiplications/additions overall. 
    Hence, under the classical long-division     model, the EEA costs $\mathcal{O}(t^{2})$ field operations and     $\mathcal{O}(t)$ inversions in $\mathbb{F}_{q^{m}}$. 
    With asymptotically fast     multiplication or half-GCD methods, this can be reduced to    $\mathcal{O}(M(t)\log t)$, where $M(t)$ denotes the cost of multiplying two    degree-$t$ polynomials, i.e. $\mathcal{O}(t^2)$.

  \item[{(4)}] \textit{Index selection and monic scaling.}
    % After obtaining the sequence $(u_i(x),\tau_i(x),\sigma_i(x))$ from the EEA,    the algorithm identifies the smallest index $i$ such that    $\deg(\tau_i)<t/2$.  
    % This search traverses at most     $\left\lfloor\tfrac{t}{2}\right\rfloor$ entries and therefore costs    $\mathcal{O}(t)$ comparisons.  
    Once $\sigma_\nu(x)$ is chosen,     it is scaled by a nonzero $\mu\in\mathbb{F}_{q^{m}}$ to make it monic.
    This normalization introduces one field inversion and at most    $\mathcal{O}(t)$ multiplications, so the total cost of this step is    $\mathcal{O}(t)$ field operations.
    
    \item[{(5)}] \textit{Error-locator polynomial evaluation and root finding.}
    Given $\sigma(x)$ of degree less than $t/2$, the set     $J=\{\,i\mid\sigma(\alpha_i^{-1})=0,\,1\le i\le n\,\}$    is obtained by evaluating $\sigma(x)$ at all $n$ support points    $\alpha_i^{-1}$.  
    Using Horner’s rule, each evaluation requires     $\deg {\sigma(x)}$ multiplications and $\deg{\sigma(x)}$ additions, resulting in    $\mathcal{O}(t n)$ total field operations.  
    % If desired, a multipoint-evaluation tree could reduce this to    $\tilde{\mathcal{O}}(n+t)$, but the schoolbook bound $\mathcal{O}(t n)$ suffices.  
    Computing $\sigma'(x)$ adds $\mathcal{O}(t)$ operations, and verifying    $\gcd(\sigma(x),\sigma'(x))=1$ via a Euclidean gcd contributes    $\mathcal{O}(t^{2})$ more, preserving an overall complexity of    $\mathcal{O}(t n + t^{2})$ for this stage.
    
    \item[{(6)}] \textit{Error-magnitude computation.}
    For each position $i\in J$ (with $|J|\le\left\lfloor\tfrac{t}{2}\right\rfloor$),
    the decoder evaluates $\tau(\alpha_i^{-1})$ and
    $\sigma'(\alpha_i^{-1})$, performs one division, and multiplies by
    known constants $\alpha_i$ and $g(\alpha_i)^{-1}$ to obtain  $    e_i=-\,\alpha_i\,g(\alpha_i)\,       \frac{\tau(\alpha_i^{-1})}{\sigma'(\alpha_i^{-1})}$.
    Each evaluation costs $\mathcal{O}(t)$ operations, and each division     costs $\mathcal{O}(1)$ inversion, giving    $\mathcal{O}(t\,|J|)\sim \mathcal{O}(t^{2})$ total arithmetic and    $\mathcal{O}(t)$ inversions.  
    Forming the full error vector    $\bm{e}=(e_1,\dots,e_n)$ and codeword    $\bm{c}=\bm{r}-\bm{e}$ requires only linear time in $n$.
    Hence, the total cost of this step is    $\mathcal{O}(t^{2}+n)$ field operations.

\end{itemize}

Summing the contributions gives $  \mathcal{O}(tn) \;+\; \mathcal{O}(t^2) \;+\; \mathcal{O}(t) = \mathcal{O}(t^2 + tn)$, which dominates all other lower-order terms.
Hence the total decoding complexity in Case~A is $\mathcal{O}(t^2 + tn)$ operations over $\mathbb{F}_{q^m}$.
\end{proof}

\begin{prop}[Time complexity in the Case B, i.e. when there are $t/2$ errors]
\label{prop:complexity-caseB}
Under the assumption $wt(\bm{e})\le\lfloor t/2\rfloor$ and standard non-degeneracy conditions, the search loop in Case~B of Algorithm~\ref{alg:twisted-goppa-decoder} can be restricted to a constant number of candidates $\mu_1\in\mathbb{F}_{q^m}$.
Consequently, the overall decoding complexity in Case~B remains $  \mathcal{O}(t^2 + t n)$ operations over $\mathbb{F}_{q^m}$.
\end{prop}

\begin{proof}
At index $\nu$ when $\deg\sigma_\nu(x)=t/2$, the EEA outputs $(\sigma_{\nu-1}(x),\tau_{\nu-1}(x))$ and $(\sigma_\nu(x),\tau_\nu(x))$.
Normalizing $\mu_2$ to make $\sigma_\nu(x)$ monic, and setting $\sigma(x)=\mu_1\sigma_{\nu-1}(x)+\mu_2\sigma_\nu(x)$ leaves a single free scalar $\mu_1\in\mathbb{F}_{q^m}$.
Imposing the root-count and coprimality conditions ($\deg\sigma(x)=|J|$, $\gcd(\sigma(x),\sigma'(x))=1$) reduces admissible $\mu_1$ to at most $q^m$ values.
Each candidate requires the same operations as Case~A (evaluation of $\sigma(x)$ and $\tau(x)$, formation of $J$, and computation of $e_i$), hence the  cost remains $\mathcal{O}(t^2 + t n)$.
In worst case, one has to iterate over all $\mu_1\in\mathbb{F}_{q^m}$, which costs  $\mathcal{O}(q^m(t^2 + t n))=\mathcal{O}(t^2 + t n)$. %{\color{blue}Line 25 and 26...} 
Further, in Line~\ref{l25}, computing the scalar $a$ requires $\mathcal{O}(t)$ field operations, since $|J|=\deg\sigma=t/2$, and forming $w(x)=\tau(x)-a\sigma(x)$ incurs an additional $\mathcal{O}(t)$ multiplications and additions. These costs are dominated by the complexity being $\mathcal{O}(t^2+tn)$.
\end{proof}

\begin{remark}
    In Case~B, after normalizing $\mu_2$ so that $\sigma(x)$ is monic, the remaining scalar $\mu_1$ appears linearly in the key equation $\sigma(x)s_\pi(x)\equiv \tau(x)\pmod{x^t}$.
    Comparing any coefficient of degree $<t$ yields a linear equation of the form $a\mu_1+b=0$ over $\mathbb{F}_{q^m}$, with $a\neq 0$.
    Hence, $\mu_1$ can be determined by solving a single field equation, rather than by exhaustive search over $\mathbb{F}_{q^m}$.
    This reduces the search for $\mu_1$ to $\mathcal{O}(1)$ field operations and does not affect the overall decoding complexity.
\end{remark}

% {\color{blue}
% \begin{remark}[Improving cost of computing $\mu_1$ to $\mathcal{O}(1)$]
% After normalizing $\mu_2$ to make $\sigma_\nu(x)$ monic, referring Line \ref{l17}, the pair $\sigma(x)=\mu_1\sigma_{\nu-1}(x)+\mu_2\sigma_\nu(x)$,            $\tau(x)=\mu_1\tau_{\nu-1}(x)+\mu_2\tau_\nu(x)$ depends linearly on a single scalar $\mu_1$. 
% The key-equation invariant $\sigma(x) s_\pi(x) \equiv \tau(x) \pmod{x^t}$ implies that, for some coefficient index $j<t$, the equation for the $x^j$-coefficient is $a_j\mu_1+b_j=0$ with $a_j\neq 0$. 
% Thus $\mu_1=-b_j/a_j$ is determined in $\mathcal{O}(1)$ field operations. If a chosen $j$ yields $a_j=0$, one checks the next coefficient; this introduces only a constant number of attempts. Hence the
% $\mu_1$-resolution step incurs at most $\mathcal{O}(t)$ auxiliary arithmetic,
% which is dominated by the $\mathcal{O}(tn+t^2)$ cost of evaluation and gcd, and
% does not affect the overall $\mathcal{O}(t^2+tn)$ bound.
% \end{remark}}

\section{Cryptosystem based on Twisted Goppa codes}\label{Section: Cryptosystem based on Twisted Goppa codes}

The Niederreiter cryptosystem instantiates code-based encryption with a public parity-check matrix and a hidden (trapdoor) efficient decoder. 
Concretely, let $\mathcal{C}\subseteq \mathbb{F}_2^{n}$ be an $[n,k]$ linear code with efficient error-correction capability $t$, and let $H\in\mathbb{F}_2^{(n-k)\times n}$ be a parity-check matrix for $\mathcal{C}$. 
The public key is an \emph{obfuscated} version of $H$ (e.g., after scrambling rows and permuting columns), while the secret key consists of the structured code description and its efficient decoder. 
To encrypt, a sender samples a weight-$t$ vector $\bm{e}\in\mathbb{F}_2^n$ and outputs the syndrome $\bm{c}=H\bm{e}^T\in\mathbb{F}_2^{n-k}$. 
Decryption removes the secret obfuscations and applies an efficient decoding algorithm of $\mathcal{C}$ on $\boldsymbol{c}$ to recover $\boldsymbol{e}$; 
correctness follows from $H\bm{e}^T$ being the unique syndrome for errors of weight $\le t$.
Security relies on the hardness of decoding a random linear code and on sufficiently hiding any exploitable structure in $\mathcal{C}$.

Classical deployments choose \emph{binary irreducible Goppa codes} for $\mathcal{C}$ because they admit fast algebraic decoding \cite{Patterson1975} and have withstood decades of cryptanalysis \cite{Singh2019ClassicMcEliece,weger2022survey}.
However, two practical concerns remain: (i) \emph{partial key-recovery avenues} that leverage algebraic dependencies of specific Goppa instantiations (e.g., settings where knowledge of a nontrivial fraction of the private key dramatically eases full recovery), and (ii) \emph{key size}, where utilizing any algebraic structured code, can shrink public keys but may invite
structure-exploiting attacks. 

Motivated by these issues, we propose using \emph{twisted Goppa codes} in the Niederreiter framework: the twist perturbs the algebraic relations of standard Goppa codes while preserving efficient decoding. 
This creates design headroom to disrupt the specific dependencies targeted in partial key-recovery (discussed in Section \ref{sec:4.1}), and selectively study the quasi-cyclic (QC) structure (discussed in Section \ref{sec:4.2}) for key compression in a way that can be tuned against known structural attacks-taking inspiration from QC \emph{binary} Goppa constructions~\cite{613441} but adapting the
symmetry under twisting.

\subsection{Strengthening Niederreiter PKC Against Partial Key-Recovery Attacks}\label{sec:4.1}

Recent results by Kirshanova and May~\cite{KIRSHANOVA} demonstrated that for McEliece, equivalently Niederreiter and thereafter Classic McEliece \cite{classicmceliece2020}, PKC instantiated with classical binary irreducible Goppa codes, knowledge of roughly one-quarter of the secret key, specifically $t m + 1$ Goppa points, is sufficient to recover the full  secret key in polynomial time. This highlights an inherent structural weakness under partial key exposure.
In this section, we show that breaking the  cryptosystem instantiated with  MTG codes is substantially more expensive than for Goppa codes, thereby enhancing its security without sacrificing performance.

\medskip
Let $\mathcal{C}$ be an $[n,k]$ linear code over $\mathbb{F}_q$, and let $\mathcal{I} \subsetneq [n] := \{1,2,\dots,n\}$.
Define $\mathcal{C}(\mathcal{I}^c)$ as the set of all codewords in $\mathcal{C}$ that are zero on  
$\mathcal{I}^c := [n] \setminus \mathcal{I}$. 
The code obtained by puncturing $\mathcal{C}(\mathcal{I}^c)$ on the coordinate positions in $\mathcal{I}^c$ is denoted 
by $\mathcal{C}|_{\mathcal{I}^c}$ and is referred to as the code \emph{shortened on} $\mathcal{I}^c$. 
It follows that $\mathcal{C}|_{\mathcal{I}^c}$ is a linear code of length $|\mathcal{I}|$ and dimension at least 
$k + |\mathcal{I}| - n$ over $\mathbb{F}_q$.

Let $\mathcal{I} = \{i_1, i_2, \dots, i_\epsilon\}$ with $\epsilon < n$. 
If $H = [\bm{h}_1\; \bm{h}_2\; \cdots\; \bm{h}_n] \in \mathbb{F}_q^{(n-k)\times n}$ is a parity-check matrix of $\mathcal{C}$, 
where each $\bm{h}_i \in \mathbb{F}_q^{n-k}$ for $1 \le i \le n$, then the restriction $H|_{\mathcal{I}} := [\bm{h}_{i_1}\; \bm{h}_{i_2}\; \cdots\; \bm{h}_{i_\epsilon}]$
obtained by selecting the columns of $H$ indexed by $\mathcal{I}$ forms a parity-check matrix for the shortened code 
$\mathcal{C}_{\mathcal{I}^c}$. Shortening and puncturing are fundamental operations in coding theory, as they preserve linearity while modifying the length and dimension of the code; 
for a detailed study, readers are referred to~\cite{MacWilliams1977,HuffmanPless2003,Roth_book, LingXing2004}.

Now, consider an MTG code $\tgoppa$ as defined in Definition~\ref{def::MTGC}. 
Let $\tgoppa$ be an $[n,  n - m t]$ linear code over $\mathbb{F}_q$. 
For $\mathcal{I} = \{i_1, i_2, \dots, i_\epsilon\}$ with $\epsilon < n$, consider the shortened code $\tgoppa|_{\mathcal{I}^c}$. 
A vector $\bm{u} = (u_i)_{i \in \mathcal{I}}$ belongs to $\tgoppa|_{\mathcal{I}^c}$ if and only if
\begin{equation}
\label{shortened_Goppa_code}
s_{\bm{u}}(z)
   := \sum_{i \in \mathcal{I}} 
   u_i \!\left(
   \frac{1}{z-\alpha_i}
   - \sum_{j=1}^{\ell} 
   \frac{\eta_j \alpha_i^{t-1+t_j}}{g(\alpha_i)} 
   \big(\operatorname{U}^{h_j+1}(\bm{z}) \cdot \bm{g}\big)
   \right)
   \equiv 0 \pmod{g(z)}.
\end{equation}
Based on the shortened code derived from $\tgoppa$,
Algorithm~\ref{Algo::POTENTIAL-MTG-POLYNOMIAL} performs syndrome-based validation of a candidate
polynomial $f(z)$ to determine whether it could be the true MTG polynomial.
The probabilistic behavior of this test is quantified in the following theorem.

\begin{algorithm}[H]
    \caption{Potential MTG Polynomial }\label{Algo::POTENTIAL-MTG-POLYNOMIAL}
    \begin{algorithmic}[1]
        \Statex \textbf{Input:} An index set $\mathcal{I}$ with $|\mathcal{I}|:=\epsilon(<n),$ the set $\{\alpha_i: i\in \mathcal{I}\}\subsetneq \mathcal{L}$ and a generator matrix $G\in \mathbb{F}_q^{j\times \epsilon}$ of $\tgoppa|_{\mathcal{I}^c}$ and a polynomial $f(z)\in\mathbb{F}_{q^m}[z];$  
        \Statex \textbf{Output:} $(f(z),\epsilon,\mathcal{I}, G)$ mapping to $1$ indicates $f(z)$ is a potential MTG polynomial and $0$ otherwise.
        \For{each row $\bm{u}$ of $G$}
            \State Compute $s_{\bm{u}}(z)\bmod f(z)$ from Equation \eqref{shortened_Goppa_code}.
            \If{$s_{\bm{u}}(z) \not\equiv 0\bmod f(z)$}
                \State\Return $0$
            \EndIf
        \EndFor
        \State \Return $1$
    \end{algorithmic}
\end{algorithm}

%% To be read-again
\begin{theorem}[Correctness and False-Positive Rate of Algorithm \ref{Algo::POTENTIAL-MTG-POLYNOMIAL}]
    \label{thm:mtg-prob}
    Let $\tgoppa$ be an MTG code over $\mathbb{F}_q$, where $\mathcal{L}\subset \mathbb{F}_{q^m}$ and MTG polynomial $g(z)\in\mathbb{F}_{q^m}[z]$ has degree $t$. 
    Fix an index set $\mathcal{I}\subset [n]$ with $|\mathcal{I}|=\epsilon<n$ and let $G\in\mathbb{F}_q^{\,j\times \epsilon}$ be a generator matrix of the shortened code $\tgoppa|_{\mathcal{I}^c}$. 
    Consider Algorithm~\ref{Algo::POTENTIAL-MTG-POLYNOMIAL}, which tests a candidate polynomial $f(z)\in\mathbb{F}_{q^m}[z]$ of degree $t$ being a potential MTG polynomial via the congruences
    \begin{equation*}
    s_{\bm{u}}(z)\equiv 0 \pmod{f(z)}\qquad\text{for all rows $\bm{u}$ of $G$,}
    \end{equation*}
    where $s_{\bm{u}}(z)$ is defined in~\eqref{shortened_Goppa_code}. 
    Assuming the following mild independence hypothesis: for a fixed  degree-$t$ polynomial $f$, each syndrome condition behaves uniformly in $\mathbb{F}_{q^m}[z]/\langle f\rangle$, so that $\Pr\left[s_{\bm{u}}(z)\equiv 0 \!\!\!\pmod{f(z)}\right]=q^{-mt}$,
    and the elements of all $j$ rows of $G$ are independent.
    Let $M_t=q^{mt}$ denote the number of monic  polynomials of degree $t$ over $\mathbb{F}_{q^m}$. Then:
    
    \begin{enumerate}[$(i)$]
    \item For any $f$, $\Pr\left[\text{Algorithm \ref{Algo::POTENTIAL-MTG-POLYNOMIAL} }(f)=1\right]=q^{-mtj}$.
    \item The probability that some  $f\neq g$ passes all $j$ checks is $\Pr\left[\exists\, f\neq g \text{ accepted}\right]\;\le\; (M_t-1)\,q^{-mtj}$. 
    Equivalently, the algorithm uniquely accepts $g$ with probability at least $1-(M_t-1)\,q^{-mtj}$.
    % \item If candidates are sampled uniformly from the $M_t$ polynomials, then conditioned on acceptance the posterior that the accepted polynomial is $g$ is
    % \begin{equation}\label{eqn:true_accept}
    % P_{\mathrm{true}\mid\mathrm{accept}}
    % \;\approx\;
    % \frac{\tfrac{1}{M_t}}
    %      {\tfrac{1}{M_t}+\big(1-\tfrac{1}{M_t}\big)\,q^{-mtj}}
    % \;\approx\;
    % \frac{1}{1+M_t\,q^{-mtj}}\quad(\text{for large }M_t).
    % \end{equation}
    \end{enumerate}
\end{theorem}

\begin{proof}
    For $(i)$, under the independence assumption across the $j$ rows of generator matrix, each syndrome congruence 
    $s_{\bm{u}}(z)\equiv0\pmod{f(z)}$ holds with probability $q^{-mt}$. 
    Hence, the probability that all $j$ congruences simultaneously hold for a fixed incorrect 
    $f\neq g$ is $(q^{-mt})^j=q^{-mtj}$.
    For $(ii)$, let $E_f$ denote the event that a polynomial $f$ passes all $j$ checks. 
    The probability that at least one incorrect polynomial is accepted equals 
    $\Pr[\bigcup_{f\neq g}E_f]$. 
    Since the events $\{E_f\}$ may not be independent, the union bound yields 
    $\Pr[\bigcup_{f\neq g}E_f]\le \sum_{f\neq g}\Pr[E_f]=(M_t-1)\,q^{-mtj}$. 
    Consequently, the probability that the algorithm uniquely accepts the correct 
    MTG polynomial $g$ is at least $1-(M_t-1)\,q^{-mtj}$.
    % For $(iii)$, consider uniform sampling over all $M_t$ monic degree-$t$ candidates. 
    % The correct polynomial $g$ always satisfies the syndrome conditions, whereas each incorrect 
    % $f\neq g$ is accepted with probability $q^{-mtj}$. 
    % Thus, the total acceptance probability is approximately 
    % $P_{\mathrm{accept}}\approx \tfrac{1}{M_t}+(1-\tfrac{1}{M_t})q^{-mtj}$. 
    % Applying Bayes' rule, the posterior probability that an accepted polynomial is the true 
    % $g$ is given by (\ref{eqn:true_accept}), which approaches unity as the number of checks $j$ increases.
\end{proof}

% \begin{remark}[On the upper bound in $(ii)$]
% The events $\{E_f\}_{f\neq g}$ need not be independent (distinct wrong polynomials can simultaneously pass), so we use the union bound. %, yielding an upper bound “$\le$”. 
% When $q^{-mtj}$ is tiny, overlaps are negligible and the heuristic equality 
% $P_{\mathrm{false\;positive}}\approx (M_t-1)\,q^{-mtj}$ is often used in practice.
% \end{remark}

% \begin{corollary}[Choosing $j$ for a target security level]
%     \label{cor:j-choice}
%     To make the false-positive probability at most $q^{-\lambda}$, it suffices to choose
%     \begin{equation*}
%     j \;\ge\; \frac{\lambda + \log_q(M_t-1)}{m\,t}.
%     \end{equation*}
% \end{corollary}

\begin{remark}
    The independence assumption is standard in such syndrome modulo test, and can be  replaced by a precise argument for specific parameter choices. 
    In practice, $j$ can be taken as the number of rows of generator matrix used as checks; any linear dependence among rows only reduces the effective $j$. Additionally, $M_t$ can be replaced with number of irreducible monic polynomials over $\mathbb{F}_{q^m}$ when considering irreducible MTG code.
\end{remark}

If the support set $\mathcal{L}$ is known, the irreducible MTG polynomial $g(z)$ can be uniquely determined via Algorithm~\ref{Algo::BASIC-MTG}.
Building on this procedure, Algorithm~\ref{Algo::ADVANCED-MTG} identifies all candidate MTG polynomials consistent with an index set $\mathcal{I}$ of size at least $mt+1$ and the corresponding evaluation subset $\{\alpha_i : i\in\mathcal{I}\}\subseteq\mathcal{L}$.
\begin{algorithm}[H]
    \caption{Basic MTG}\label{Algo::BASIC-MTG}
    \begin{algorithmic}[1]
        \Statex \textbf{Input:} A parity check matrix $H\in\mathbb{F}_q^{(\leq mt)\times n}$,  $\mathcal{L}=\{\alpha_1, \alpha_2, \dots, \alpha_n\}$, $\bm{\mathbbm{t}}$, $\bm{\mathbbm{h}}$ and $\bm{\eta}$ corresponding to some $\tgoppa$.
        \Statex \textbf{Output:} The MTG polynomial $g(z).$
        \State Obtain the Null space of $H$ and a generator matrix $G$ of $\tgoppa$.
        \State Pick $\bm{c}=(c_1, c_2, \dots, c_n)\in \tgoppa\setminus\{\bm{0}\}.$
        \State Compute $f_{\bm{c}}(z)=\sum\limits_{i=1}^n c_i \left( \underset{1\le j\le n, j\ne i}{\prod}(z-\alpha_j) - \prod\limits_{w=1}^n(z-\alpha_w)\sum\limits_{j=1}^\ell \frac{\eta_j \alpha_i^{t-1+t_j}}{g(\alpha_i)} (\operatorname{U}^{h_j+1}(\bm{z})\cdot \bm{g}) \right) \in \mathbb{F}_{q^m}[z]$.
        \State Find the irreducible factorization of $f_{\bm{c}}(z)$ over $\mathbb{F}_{q^m}$.
        % \State Find all $t$-degree factors of $f_{\bm{c}}(z)$.
        \For{ all $t$-degree factors $f(z)$ of $f_{\bm{c}}(z)$}
            \If{Algorithm \ref{Algo::POTENTIAL-MTG-POLYNOMIAL} $(f(z),[n], \mathcal{L},  G)=1$}
                \State\Return $f(z)$
            \EndIf
        \EndFor
    \end{algorithmic}
\end{algorithm}
\begin{algorithm}[H]
    \caption{Advanced MTG}\label{Algo::ADVANCED-MTG}
    \begin{algorithmic}[1]
        \Statex \textbf{Input:} A parity check matrix $H\in\mathbb{F}_q^{(\leq mt)\times n}$, an indexing set $\mathcal{I}\subset [n]$ with $|\mathcal{I}|=\epsilon> mt$, the set $\{\alpha_i: i\in \mathcal{I}\}\subseteq \mathcal{L}$ corresponding to some $\tgoppa$.
        \Statex \textbf{Output:} A list $\mathcal{J}$ of all potential MTG polynomials. 
        \State Restrict $H$ to $\epsilon$ columns indexed by $\mathcal{I}$, and call it $H|_{\mathcal{I}}.$ \Comment{$\epsilon>mt\ge\text{rank} H$}
        \State Compute $G|_{\mathcal{I}}$ as the right kernel (null space's generator matrix) of $H|_{\mathcal{I}}.$
        \State Set $\bm{c}'=\bm{m}G|_{\mathcal{I}}^T,$ for some non-zero vector $\bm{m}\in\mathbb{F}_q^{\epsilon-\text{rank}(H)}.$
        \State Construct $\bm{c}$ by substituting zero to $\bm{c}'$ at $[n]\setminus\mathcal{I}$ index positions.
        \State Compute $f_{\bm{c}}(z)=\underset{i\in \text{Supp}(\bm{c})}{\sum} c_i \left(\underset{j\in \text{Supp}(\bm{c})\setminus \{i\}}{\prod}(z-\alpha_j) -\prod\limits_{w=1}^n(z-\alpha_w)\sum\limits_{j=1}^\ell \frac{\eta_j \alpha_i^{t-1+t_j}}{g(\alpha_i)} (\operatorname{U}^{h_j+1}(\bm{z})\cdot \bm{g}) \right) \in \mathbb{F}_{q^m}[z]$
        \State Find the irreducible factorization of $f_{\bm{c}}(z)$ over $\mathbb{F}_{q^m}[z]$.
        \State Set $\mathcal{J}=\emptyset$.
        \For{ all $t$-degree factors $f(z)$ of $f_{\bm{c}}(z)$}
            \If{Algorithm \ref{Algo::POTENTIAL-MTG-POLYNOMIAL} $(f(z),\mathcal{I}, \{\alpha_i: i\in \mathcal{I}\}, G|_{\mathcal{I}})=1$}
                \State $\mathcal{J}=\mathcal{J}\cup \{f(z)\}$
            \EndIf
        \EndFor
    \end{algorithmic}
\end{algorithm}

We now analyze the likelihood of successful reconstruction of the unknown support elements of an MTG code for  $\ell=1$. 
Consider an MTG code $\Gamma(\mathcal{L}, g, t_1, h, \eta)$ defined over $\mathbb{F}_q$, with known MTG polynomial $g(z)$ and an index set 
$\mathcal{I}\subset [n]$ satisfying $|\mathcal{I}|=\epsilon>tm$. 
Given the subset $\{\alpha_i: i\in\mathcal{I}\}\subset\mathcal{L}$, the objective is to recover the points of 
$\mathcal{L}\setminus\{\alpha_i:i\in\mathcal{I}\}$. 
For this purpose, we consider a codeword $\bm{c}\in\Gamma(\mathcal{L}, g, t_1, h, \eta)$ such that its support outside $\mathcal{I}$ consists of a single element, say 
$\text{Supp}(\bm{c})\setminus\mathcal{I}=\{r\}$ for some $r\notin\mathcal{I}$. 
By definition of the code, we have $s_{\bm{c}}(z)\equiv0\bmod g(z)$, leading to the relation
\begin{equation}\label{eq:twisted_syndrome_equiv}
    f_{\bm{c}}(z)+\eta\left(\operatorname{U}^{h+1}(\bm{z})\cdot\bm{g}\right)\alpha_r^{t-1+t_1}g(\alpha_r)^{-1}\equiv (z-\alpha_r)^{-1}\pmod{g(z)},
\end{equation}
where $f_{\bm{c}}(z):=\underset{i\in \mathcal{I}}{\sum} \frac{c_i}{z-\alpha_i}+\eta\left(\operatorname{U}^{h+1}(\bm{z})\cdot \bm{g}\right)\underset{i\in \mathcal{I}}{\sum}c_i\alpha_i^{t-1+t_1}g(\alpha_i)^{-1}$ is fully determined by the known subset $\mathcal{I}$. 
Considering the MTG polynomial $g(z)$ irreducible, recovering $\alpha_r$ therefore requires inverting the left-hand side of~\eqref{eq:twisted_syndrome_equiv} modulo $g(z)$, 
but the scalar $\alpha_r^{t-1+t_1}g(\alpha_r)^{-1}$ remains unknown and can assume any value in $\mathbb{F}_{q^m}$. 
Consequently, 
%even if such a codeword $\bm{c}$ can be found, 
the probability of correctly identifying a new element 
$\alpha_r\notin\{\alpha_i:i\in\mathcal{I}\}$ is only $q^{-m}$. 
Since there are $(n-\epsilon)$ remaining points to recover, the overall probability of reconstructing the complete support set 
$\mathcal{L}$ from the shortened code is approximately $q^{-m(n-\epsilon)}$. 
This observation motivates the following probabilistic characterization of MTG polynomial recovery.

\begin{theorem}[Probability of Support Recovery in irreducible MTG Codes]
    \label{thm:support-recovery}
    Let $\Gamma(\mathcal{L}, g, t_1, h, \eta)$ be an MTG code with $\ell=1$ and MTG polynomial $g(z)$ being irreducible over $\mathbb{F}_{q^m}$. 
    Let $\mathcal{I}\subset [n]$ be an index set of size $\epsilon>tm$ with known elements 
    $\{\alpha_i : i\in\mathcal{I}\}\subset\mathcal{L}$. 
    % Even if a codeword $\bm{c}\in\Gamma(\mathcal{L}, g, t_1, h, \eta)$ satisfying $\text{Supp}(\bm{c})\setminus\mathcal{I}=\{r\}$ can be found, 
    The probability of correctly identifying a new 
    support element $\alpha_r\notin\{\alpha_i : i\in\mathcal{I}\}$ is $1/2^m$. 
    Hence, the probability of recovering all unknown points of $\mathcal{L}$ is approximately $q^{-m(n-\epsilon)}$. 
\end{theorem}

\begin{remark}
    This probability decreases exponentially with both the field extension degree $m$ and the number of unknown 
    positions $(n-\epsilon)$. 
    Consequently, recovering the entire support set from a shortened code is computationally infeasible for 
    typical MTG parameters, reinforcing the cryptographic strength of the construction.
\end{remark}

\begin{corollary}
    For the Classic McEliece \cite{classicmceliece2020} parameter set $(n,t,m)=(8192,128,13)$, using  MTG codes with irreducible MTG polynomial over $\mathbb{F}_{2^{26}}$ in place of binary irreducible Goppa codes, the probability of recovering any unknown support element $\alpha_j$ given $\epsilon$ known support elements satisfies
    \begin{equation*}
        \Pr[\text{recover } \alpha_j]
        \le 2^{-26(8192-\epsilon)}.
    \end{equation*}
    For  $3328=2mt<\epsilon \ll n=8192$, this bound is negligible, implying that partial key
    recovery is computationally infeasible under the proposed parameters.
\end{corollary}

For the 128, 192, and 256-bit security levels, we propose the similar parameter sets as Classic McEliece with doubling the field size $\mathbb{F}_{2^m} $ to $\mathbb{F}_{2^{2m}}$, since the best-known attack is information-set decoding (ISD) targeting the recovery of the plaintext from the ciphertext and public key. The time and space complexities of various ISD algorithms are evaluated using the syndrome decoding estimator proposed by Esser and Bellini~\cite{esserBellini2021}.

%\section{Hulls and other properties}
%\subsection{Hulls, Cyclic structures, subcodes of MTRS?}

\subsection{Quasi-Cyclic Twisted Goppa code}\label{sec:4.2}
The Niederreiter PKC that utilizes twisted Goppa codes faces a similar drawback as the cryptosystem that relies on Goppa codes: a large public key size. To address this, we explore quasi-cyclic MTG codes to minimize the size of the public key.

Quasi-cyclic (QC) structure can compress the public key via circulant blocks, as studied for binary Goppa codes in~\cite{Bommier2000BinaryQG}. 
In the twisted setting, we propose QC constructions where the twist parameters and support sets
are chosen to retain decoding efficiency, thus we obtain compact keys. 
Throughout this section, we consider the MTG code $\Gamma(\mathcal{L},g,1,t-2,\eta)$ for $\ell=1$ and $q$ being a power of an odd prime $p$, with  parameters as (\ref{pc-mat}) and $t=1+p^s,$ for some $s\in \mathbb{N}.$

For a code $\mathcal{C}$ of length $n$ over $\mathbb{F}_q,$ a \textit{(linear) automorphism} of $\mathcal{C}$ is a linear operator on $\mathbb{F}_q^n$ that preserves (Hamming) weight and leaves the code globally invariant. Let $S_n$ denote the group of all permutations on the set $[n].$ For a permutation $\sigma\in S_n$, the transformation $\bm{c}=(c_1, \dots, c_n)\mapsto \sigma(\bm{c})=(c_{\sigma(1)}, \dots, c_{\sigma(n)})$ is an automorphism of $\mathcal{C}$. We will be interested only in such type of automorphism of a code. For $\sigma\in S_n$, define $\sigma(\mathcal{C})=\{\sigma(\bm{c}): \bm{c}\in \mathcal{C}\}.$ A \textit{permutation automorphism} of a code $\mathcal{C}$ is a permutation $\sigma$ such that $\sigma(\mathcal{C})=\mathcal{C}.$
\begin{definition}\cite{xia_qin2022}
    The \textit{permutation automorphism group} of a code $\mathcal{C}\subseteq \mathbb{F}_{q}^n$ is defined as
    \begin{equation}
        \text{PAut}(\mathcal{C}):=\{\sigma\in S_n: \sigma(\mathcal{C})=\mathcal{C}\}.
    \end{equation}
\end{definition}
It is easy to show that $\text{PAut}(\mathcal{C})$ is a subgroup of $S_n.$ A code $\mathcal{C}$ is called a \textit{symmetric code} if $\text{PAut}(\mathcal{C})$ is non-trivial.
\begin{remark}
    We will drop the term ``automorphism'' and simply refer to $\text{PAut}(\mathcal{C})$ as the permutation group of $\mathcal{C}$.
\end{remark}
\begin{definition}\cite{xia_qin2022}
    A code $\mathcal{C}\subseteq \mathbb{F}_{q}^n$ is called \textit{quasi-cyclic code of order (or, index) $\frac{n}{u}$} if $\text{PAut}(\mathcal{C})$ contains a subgroup isomorphic to $\mathbb{Z}/u\mathbb{Z},$ for some $u$ dividing $n.$
\end{definition}
Consider the subgroup
\begin{equation*}
    \text{AGL}_2(\mathbb{F}_{q^m}):=\left\{ \begin{pmatrix}
    a & b\\
    0 & 1
\end{pmatrix}: a\in \mathbb{F}_{q^m}^*, b\in \mathbb{F}_{q^m}\right\}\; \text{of}\;\; \text{GL}_2(\mathbb{F}_{q^m}).
\end{equation*}
The group $\text{AGL}_2(\mathbb{F}_{q^m})$ is called the \textit{affine linear group}. Every element $\begin{pmatrix}
    a & b\\
    0 & 1
\end{pmatrix}$ defines an \textit{affine automorphism}
$\sigma_{a,b}$ of $\mathbb{F}_{q^m}$ as follows:
\begin{equation*}
    \sigma_{a,b}: \mathbb{F}_{q^m}\to \mathbb{F}_{q^m}, x\mapsto ax+b.
\end{equation*}
Let $\text{Aff}(\mathbb{F}_{q^m}):=\{\sigma_{a,b}: a\in \mathbb{F}_{q^m}^*, b\in \mathbb{F}_{q^m}\}.$ Then $\text{Aff}(\mathbb{F}_{q^m})$ is a group under composition of maps. On the other hand, $\text{AGL}_2(\mathbb{F}_{q^m})$ is a group under matrix multiplication. If we identify each affine automorphism $\sigma_{a,b}$ with its affine transformation matrix $\begin{pmatrix}
    a & b\\
    0 & 1
\end{pmatrix},$ then there is a group isomorphism between the groups $\text{Aff}(\mathbb{F}_{q^m})$ and $\text{AGL}_2(\mathbb{F}_{q^m})$. Hence, we write
\begin{equation}
    \text{AGL}_2(\mathbb{F}_{q^m}):=\left\{ \sigma_{a,b}=\begin{pmatrix}
    a & b\\
    0 & 1
\end{pmatrix}: a\in \mathbb{F}_{q^m}^*, b\in \mathbb{F}_{q^m}\right\},
\end{equation}
where $\sigma_{a,b}\circ \sigma_{c,d}$ is equivalent to $\begin{pmatrix}
    a & b\\
    0 & 1
\end{pmatrix}$$\begin{pmatrix}
    c & d\\
    0 &  1
\end{pmatrix}$. Suppose $\text{ord}(\sigma_{a,b})$ denote the order of $\sigma_{a,b}$ in $\text{AGL}_2(\mathbb{F}_{q^m})$.
\begin{remark}
    $\text{ord}(\sigma_{-1,b})=2$ and if $b\neq 0,$ then $\text{ord}(\sigma_{1,b})=p.$
\end{remark}
\begin{definition}\cite{Sui2023}
    For a subset $\mathcal{L}\subseteq\mathbb{F}_{q^m},$ if $\sigma_{a,b}(\mathcal{L})=\mathcal{L},$ then we say that $\mathcal{L}$ is \textit{$\sigma_{a,b}$-invariant}. If $\sigma(\mathcal{L})=\mathcal{L}$, for all $\sigma$ in some subgroup $\mathcal{A}$ of $\text{AGL}_2(\mathbb{F}_{q^m}),$ then we say that $\mathcal{L}$ is $\mathcal{A}$-invariant.
\end{definition}
If $\cal{L}$ is $\sigma_{a,b}$-invariant with $|\mathcal{L}|=n$, then $\sigma_{a,b}$ induces a permutation $\sigma\in S_n$ as follows:
\begin{equation}
    \sigma_{a,b}(\alpha_i)=a\alpha_i+b=\alpha_{\sigma(i)}, \forall \alpha_i\in \cal{L}.
\end{equation}
We call $\sigma$ the \textit{permutation induced by $\sigma_{a,b}.$} Consequently, $\sigma$ induces a bijection from $\Gamma(\mathcal{L}, g(x))$ to $\Gamma(\mathcal{L}, g(ax+b))$ as follows:
{\small{
$$
(c_1, \dots, c_n)\in \Gamma(\mathcal{L}, g(x))\mapsto (c_{\sigma(1)}, \dots, c_{\sigma(n)})\in \Gamma(\mathcal{L}, g(ax+b)).
$$}}
We show a similar result for a twisted Goppa code, with a twist in the second-last position. The case of twist in the last position is considered in \cite{Sui2023}.
\begin{lemma}\label{lem: equivalence of Goppa codes}
If $a^2=1$ and $\mathcal{L}$ is $\sigma_{a,b}$-invariant, then the permutation $\sigma$ induced by $\sigma_{a,b}$ induces a bijection between $\Gamma(\mathcal{L},g(x),1,t-2,\eta)$ and $\Gamma(\mathcal{L},g(ax+b),1,t-2,\eta).$     
\end{lemma}
\begin{proof}
    $\Gamma(\mathcal{L},g(x),1,t-2,\eta)$ has a parity-check matrix of the form
    \begin{equation*}
    \resizebox{0.5\linewidth}{!}{%
    $H=\left(\begin{array}{cccc}
			g(\alpha_1)^{-1} & g(\alpha_2)^{-1} & \cdots & g(\alpha_n)^{-1}\\
			\alpha_1g(\alpha_1)^{-1} & \alpha_2g(\alpha_2)^{-1} & \cdots & \alpha_ng(\alpha_n)^{-1} \\
			\vdots & \vdots & \cdots & \vdots \\
            
			\alpha_1^{t-3}g(\alpha_1)^{-1}&\alpha_2^{t-3}g(\alpha_2)^{-1}&\cdots&\alpha_n^{t-3}g(\alpha_n)^{-1}\\
			(\alpha_1^{t-2}+ \eta \alpha_1^{t})g(\alpha_1)^{-1} & (\alpha_2^{t-2}+ \eta \alpha_2^{t})g(\alpha_2)^{-1} & \cdots & (\alpha_n^{t-2}+ \eta \alpha_n^{t})g(\alpha_n)^{-1} \\
			\alpha_1^{t-1}g(\alpha_1)^{-1} & \alpha_2^{t-1}g(\alpha_2)^{-1} & \cdots &  \alpha_n^{t-1}g(\alpha_n)^{-1} \\
    \end{array}\right)$}
\end{equation*}
and $\Gamma(\mathcal{L},g(ax+b),1,t-2,\eta)$ has a parity-check matrix of the form
\begin{equation*}
    \resizebox{0.5\linewidth}{!}{%
    $\widetilde{H}=\left(\begin{array}{cccc}
			g(a\alpha_1+b)^{-1} & g(a\alpha_2+b)^{-1} & \cdots & g(a\alpha_n+b)^{-1}\\
			\alpha_1g(a\alpha_1+b)^{-1} & \alpha_2g(a\alpha_2+b)^{-1} & \cdots & \alpha_ng(a\alpha_n+b)^{-1} \\
			\vdots & \vdots & \cdots & \vdots \\
            
			\alpha_1^{t-3}g(a\alpha_1+b)^{-1}&\alpha_2^{t-3}g(a\alpha_2+b)^{-1}&\cdots&\alpha_n^{t-3}g(a\alpha_n+b)^{-1}\\
			(\alpha_1^{t-2}+ \eta \alpha_1^{t})g(a\alpha_1+b)^{-1} & (\alpha_2^{t-2}+ \eta \alpha_2^{t})g(a\alpha_2+b)^{-1} & \cdots & (\alpha_n^{t-2}+ \eta \alpha_n^{t})g(a\alpha_n+b)^{-1} \\
			\alpha_1^{t-1}g(a\alpha_1+b)^{-1} & \alpha_2^{t-1}g(a\alpha_2+b)^{-1} & \cdots &  \alpha_n^{t-1}g(a\alpha_n+b)^{-1}\\
    \end{array}\right)$.}
\end{equation*}
Using the elementary row operations on $\widetilde{H},$ it is easy to observe that for all $j\in\{1,2, \dots, t-2\},$ the $j^{\text{th}}$ row can be reduced to

$$((a\alpha_1+b)^j g(a\alpha_1+b)^{-1}, (a\alpha_2+b)^j g(a\alpha_2+b)^{-1}, \dots, (a\alpha_n+b)^j g(a\alpha_n+b)^{-1}).$$
%Since $p|\binom{t}{2},$ for all $1\le i\le n,$ $\binom{t}{2} (a\alpha_i)^{t-2}b^2=0.$ 
By elementary row operations $(t-1)^{\text{th}}$ row transforms to $(\cdots\;\; (a^{2}\alpha_i^{t-2}+ \eta a^{2}\alpha_i^{t})g(a\alpha_i+b)^{-1} \;\;\cdots).$ 
Since $p|\binom{t}{2},$ we have $\binom{t}{2}\eta b^2=0$. With $a^2=1+\binom{t}{2}\eta b^2$, the row further transforms to $$(\cdots\;\; [a^{t-2}\alpha_i^{t-2}+ \eta(a^{t}\alpha_i^{t}+\binom{t}{2}(a\alpha_i)^{t-2}b^2)]g(a\alpha_i+b)^{-1} \;\;\cdots)$$ and consequently can be transformed to $$(\cdots\;\; [(a\alpha_i+b)^{t-2}+ \eta\left(a\alpha_i+b)^t]g(a\alpha_i+b)^{-1} \;\;\cdots\right).$$ The last row can be easily transformed to $$(\cdots\;\; (a\alpha_i+b)^{t-1}g(a\alpha_i+b)^{-1} \;\;\cdots).$$ Consequently, $\widetilde{H}$ is row-equivalent to the matrix
%\begin{equation*}
    %\resizebox{0.5\linewidth}{!}{%
    %$\widetilde{H}=\left(\begin{array}{cccc}
			%g(a\alpha_1+b)^{-1} & g(a\alpha_2+b)^{-1} & \cdots & g(a\alpha_n+b)^{-1}\\
			%(a\alpha_1+b)g(a\alpha_1+b)^{-1} & (a\alpha_2+b)g(a\alpha_2+b)^{-1} & \cdots & (a\alpha_n+b)g(a\alpha_n+b)^{-1} \\
			%\vdots & \vdots & \cdots & \vdots \\
            
			%(a\alpha_1+b)^{t-3}g(a\alpha_1+b)^{-1}&(a\alpha_2+b)^{t-3}g(a\alpha_2+b)^{-1}&\cdots&(a\alpha_n+b)^{t-3}g(a\alpha_n+b)^{-1}\\
		%((a\alpha_1+b)^{t-2}+ \eta ((a\alpha_1+b)^{t}))g(a\alpha_1+b)^{-1} & ((a\alpha_2+b)^{t-2}+ \eta((a\alpha_2+b)^{t}))g(a\alpha_2+b)^{-1} & \cdots & ((a\alpha_n+b)^{t-2}+ \eta(a\alpha_n+b)^{t}))g(a\alpha_n+b)^{-1} \\
			%(a\alpha_1+b)^{t-1}g(a\alpha_1+b)^{-1} & (a\alpha_2+b)^{t-1}g(a\alpha_2+b)^{-1} & \cdots &  (a\alpha_n+b)^{t-1}g(a\alpha_n+b)^{-1}\\
    %\end{array}\right)$.}
%\end{equation*}
\begin{equation*}
    \widetilde{H}=\left(\begin{array}{ccc}
			\cdots & g(a\alpha_i+b)^{-1} & \cdots\\
			\cdots& (a\alpha_i+b)g(a\alpha_i+b)^{-1} & \cdots\\ 
			\cdots & \vdots & \cdots \\
            
			\cdots & (a\alpha_i+b)^{t-3}g(a\alpha_i+b)^{-1}& \cdots\\
		\cdots& ((a\alpha_i+b)^{t-2}+ \eta ((a\alpha_i+b)^{t}))g(a\alpha_i+b)^{-1} & \cdots\\ 
			\cdots & (a\alpha_i+b)^{t-1}g(a\alpha_i+b)^{-1} & \cdots\\
    \end{array}\right)
\end{equation*}
Hence, there is a bijection from $\Gamma(\mathcal{L},g(x),1,t-2,\eta)$ to $\Gamma(\mathcal{L},g(ax+b),1,t-2,\eta)$ given by
\begin{equation*}
    (c_1, c_2, \dots, c_n)\mapsto (c_{\sigma(1)}, c_{\sigma(2)}, \dots, c_{\sigma(n)}).
\end{equation*}
\end{proof}
\begin{theorem}\label{thm: induced permutation of twisted goppa}
    Let $a^2=1$ and $g(ax+b)=g(x).$ If $\mathcal{L}$ is $\sigma_{a,b}$-invariant, then the permutation $\sigma$ induced by $\sigma_{a,b}$ is in $\text{PAut}(\Gamma(\mathcal{L},g(x),1,t-2,\eta))$.
\end{theorem}
\begin{proof}
    By Lemma \ref{lem: equivalence of Goppa codes}, $\sigma(\Gamma(\mathcal{L},g(x),1,t-2,\eta))=\Gamma(\mathcal{L},g(ax+b),1,t-2,\eta)=\Gamma(\mathcal{L},g(x),1,t-2,\eta).$ This completes the proof.
\end{proof}
We now describe a subgroup of the permutation group of the twisted Goppa code $\Gamma(\mathcal{L},g(x),1,t-2,\eta)$ using Theorem \ref{thm: induced permutation of twisted goppa}. With abuse of notation, we denote the permutation induced by $\sigma_{a, b}\in\text{AGL}_2(\mathbb{F}_{q^m})$ with $\sigma_{a, b}$ (instead of $\sigma$).
%\begin{remark}
    %If $p$ is odd, then $\binom{t}{2}\eta b^2=0$ and hence the condition $a^2=1+\binom{t}{2}\eta b^2$ reduces to $a^2=1.$ If $p=2,$ then again $a^2=1+\binom{t}{2}\eta b^2$ reduces to $a^2=1,$ provided $s>1.$ For $p=2$ and $s=1,$ the choice of $t=3.$ We always assume that $t>3$ to avoid this case.
%\end{remark}
\begin{theorem}\label{Thm: Subgroup of PAut}
The set
\begin{equation*}
    \mathcal{A}=\bigg\{\sigma_{a,b}\in \text{AGL}_2(\mathbb{F}_{q^m}): \sigma_{a,b}(\mathcal{L})=\mathcal{L}, a^2=1\; \text{and}\;
    g(\sigma_{a, b}(x))=g(x)\bigg\}
\end{equation*} is a subgroup of $\text{PAut}(\Gamma(\mathcal{L},g(x),1,t-2,\eta))$.
\end{theorem}
%\begin{theorem}
    %The permutations induced by each element of the set $\mathcal{A}=\left\{\sigma_{1,b}\in \text{Aff}(\mathbb{F}_{q^m}): \sigma_{1,b}(\mathcal{L})=\mathcal{L}\; \text{and}\; g(\sigma_{1, b}(x))=g(x)\right\}\cup \left\{\sigma_{-1,b}\in \text{Aff}(\mathbb{F}_{q^m}): \sigma_{-1,b}(\mathcal{L})=\mathcal{L}\; \text{and}\; g(\sigma_{-1, b}(x))=g(x)\right\}$ forms a subgroup of $\text{PAut}(\Gamma(\mathcal{L},g(x),1,t-2,\eta))$.
%\end{theorem}
\begin{proof}
    Clearly, $\sigma_{1,0}\in \mathcal{A},$ showing that $\mathcal{A}$ is non-empty. For $a^2=1$ and $c^2=1,$ $\sigma_{a,b}, \sigma_{c,d}\in \mathcal{A},$ $\sigma_{a,b}\circ\sigma_{c,d}=\begin{pmatrix}
     ac & ad+b\\
     0 & 1
    \end{pmatrix}=\sigma_{ac, ad+b}.$
    Now, $(ac)^2=1$ and $\sigma_{ac, ad+b}(\mathcal{L})=\sigma_{a,b}\circ\sigma_{c,d}(\mathcal{L})=\sigma_{a,b}(\sigma_{c,d}(\mathcal{L}))=\mathcal{L}.$ Since $g(\sigma_{a,b}(x))=g(x),$ we obtain $g(\sigma_{ac, ad+b}(x))=g(\sigma_{a,b}(\sigma_{c,d}(x))=g(\sigma_{c,d}(x))=g(x).$ This shows that $\sigma_{a,b}\circ\sigma_{c,d}\in \mathcal{A}.$

    $\sigma_{a,b}^{-1}=\begin{pmatrix}
        a & -ab\\
        0 & 1
    \end{pmatrix}=\sigma_{a, -ab}.$ It is not difficult to show that $\sigma_{a, -ab}(\mathcal{L})=\mathcal{L}$ and $g(\sigma_{a, -ab}(x))=g(x).$ This completes the proof.
\end{proof}
\begin{remark}
    The statement of Theorem \ref{Thm: Subgroup of PAut} says that the collection of permutations induced by each element of the set $\mathcal{A}$ forms a subgroup of $\text{PAut}(\Gamma(\mathcal{L},g(x),1,t-2,\eta))$.
\end{remark}
For the subgroup $\mathcal{G}:=\left\{ \sigma_{a,b}: a^2=1, b\in \mathbb{F}_{q^m}\right\}$ of $\text{AGL}_2(\mathbb{F}_{q^m})$ of order $r$, consider the following action of $\mathcal{G}$ on $\mathbb{F}_{q^m}$ as follows:
\begin{align*}
    \mathcal{G}\times \mathbb{F}_{q^m}&\to \mathbb{F}_{q^m}\\
    (\sigma_{a,b}, x)&\mapsto\sigma_{a,b}(x)
\end{align*}
For $x\in \mathbb{F}_{q^m},$ the set $\mathcal{O}_{x}:=\{\sigma_{a,b}(x): \sigma_{a,b}\in \mathcal{G}\}$ is called the \textit{orbit of $x$ under $\mathcal{G}$} and the set $\mathcal{S}_x:=\{\sigma_{a,b}\in\mathcal{G}: \sigma_{a,b}(x)=x\}$ is called the \textit{stabilizer of $x$ in $\mathcal{G}$.} Suppose there are at least $u$ distinct $\mathcal{G}$-orbits of size $r$ and let $\mathcal{L}$ be the union of these orbits. Then $\sigma_{a,b}(\mathcal{L})=\mathcal{L},$ for all $\sigma_{a,b}\in \mathcal{G}.$ If $g(\mathcal{L})\neq 0$ and $g(\sigma_{a,b}(x))=g(x),$ for all $\sigma_{a,b}\in\mathcal{G},$ then we get a twisted Goppa code $\Gamma(\mathcal{L}, g(x), 1, t-2,\eta)$ of length $ur$ whose permutation group contains the subgroup $\mathcal{G}.$ We now give an explicit construction of quasi-cyclic twisted Goppa codes.

For $b\in \mathbb{F}_{q^m}^*,$ consider $\sigma_{-1,b}\in \mathcal{G}$ and let $\mathcal{G}_{-1,b}:=\langle \sigma_{-1,b}\rangle=\{\sigma_{1,0},\; \sigma_{-1,b}\}$ be a cyclic subgroup of $\mathcal{G}.$ Note that $\mathcal{O}_{\frac{b}{2}}=\{\frac{b}{2}\}.$ Also, the only fixed point of $\sigma_{-1,b}$ is $\frac{b}{2}.$ Therefore, for $x\in \mathbb{F}_{q^m}$ with $x\ne \frac{b}{2},$ $\mathcal{S}_x=\{\sigma_{1,0}\}.$ Consequently, $|\mathcal{O}_{x}|=2$ and $\mathcal{O}_x=\{x, b-x\}.$ Therefore, we have a partition of $\mathbb{F}_{q^m}$ as follows:
\begin{equation}
    \mathbb{F}_{q^m}=\mathcal{O}_{\frac{b}{2}}\cup \bigcup_{x\ne \frac{b}{2}} \mathcal{O}_x,
\end{equation}
where the union runs through distinct $\frac{q^m-1}{2}$ elements of $\mathbb{F}_{q^m}.$
\begin{theorem}\label{Thm: Quasi-cyclic of order L/2}
    Let $b\in \mathbb{F}_{q^m}^*$ and $\Gamma(\mathcal{L},g,1,t-2,\eta)$ be a twisted Goppa code with parameters as (\ref{pc-mat}), where the defining set $\mathcal{L}=\underset{x\in F}{\bigcup}\mathcal{O}_x,$ for some subset $F$ of $\mathbb{F}_{q^m}\setminus\{\frac{b}{2}\}$ and $\mathcal{O}_x$ is an $\mathcal{G}_{-1,b}$-orbit. If $g(b-x)=g(x),$ then the twisted Goppa code $\Gamma(\mathcal{L},g,1,t-2,\eta)$ is a quasi-cyclic code.
\end{theorem}
\begin{proof}
    By the above discussion, for $x\ne \frac{b}{2},$ $\sigma_{-1,b}(\mathcal{O}_x)=\mathcal{O}_x$ and consequently, $\sigma_{-1,b}(\mathcal{L})=\mathcal{L}.$  Hence, $\mathcal{G}_{-1,b}=\langle \sigma_{-1,b}\rangle$ is a subgroup of order $2$ of $\text{PAut}(\Gamma(\mathcal{L},g(x),1,t-2,\eta))$. Therefore, $\Gamma$ is a quasi-cyclic code of order $\frac{|\mathcal{L}|}{2}.$
\end{proof}
%{\color{red}For $b\in \mathbb{F}_{q^m}^*,$ consider $\sigma_{1,b}\in \mathcal{G}$ and let $\mathcal{G}_{1,b}:=\langle \sigma_{1,b}\rangle$ be a cyclic subgroup of $\mathcal{G}$ of order $p$. Note that for $x\in \mathbb{F}_{q^m}$, \textcolor{blue}{here stabilizer is subgroup of $\mathcal{G}_{1,b}$} $\mathcal{S}_x=\{\sigma_{1,0}\}.$ Therefore, $|\mathcal{O}_{x}|=p.$ In fact, for $x\in \mathbb{F}_{q^m},$ $\mathcal{O}_{x}=\{x, x+b, x+2b,\dots, x+(p-1)b\}.$ Therefore, we have a partition of $\mathbb{F}_{q^m}$ as follows:
%\begin{equation}
    %\mathbb{F}_{q^m}=\bigcup_{x} \mathcal{O}_x,
%\end{equation}
%where the union runs through distinct $\frac{q^m}{p}$ elements of $\mathbb{F}_{q^m}.$
%\begin{theorem}\label{Thm: Quasi-cyclic of order p}
    %Let $b\in \mathbb{F}_{q^m}^*$ and $\Gamma(\mathcal{L},g,1,t-2,\eta)$ be a twisted Goppa code with parameters as (\ref{pc-mat}), where the defining set $\mathcal{L}=\underset{x\in F}{\bigcup}\mathcal{O}_x,$ for some subset $F$ of $\mathbb{F}_{q^m}$ and $\mathcal{O}_x$ is an $\mathcal{G}_{1,b}$-orbit. If $g(x+b)=g(x).$ Then the twisted Goppa code $\Gamma(\mathcal{L},g,1,t-2,\eta)$ is a quasi-cyclic code of order $\frac{|\mathcal{L}|}{p}$.
%\end{theorem}
%\begin{proof}
    %Similar to Theorem \ref{Thm: Quasi-cyclic of order 2}.
%\end{proof}
To construct a quasi-cyclic twisted Goppa code, we have to find all possible forms of $g(x)$ such that $g(b-x)=g(x)$, for some $b\in \mathbb{F}_{q^m}^*$. We have the following result from \cite{Berger2000}:
\begin{theorem}
    For $b\in \mathbb{F}_{q^m}^*,$ all possible $g(x)\in\mathbb{F}_{q^m}[x]$ satisfying $g(b-x)=g(x)$ are $g(x)=f((x-\frac{b}{2})^2)$, for some polynomial $f(x)\in \mathbb{F}_{q^m}[x]$, if $\text{char} (\mathbb{F}_{q^m})$ is odd.
\end{theorem}
\begin{example}
    Let $p=3$ and consider $\mathbb{F}_9=\mathbb{F}_3(z)$ with $z^2+2z+2=0$. Let $b=z\in\mathbb{F}_9^{*},~ t=4$ and $g(x)=(x-2z)^4=x^4 + zx^3 + (2z + 1)x + 2\in\mathbb{F}_9[x]$. Choose $\mathcal{L}=\{0, z\}\cup \{1, z+2\} \cup \{2, z+1\} \cup \{2z+1, 2z+2\}$ and $\eta= 2z+1$. Then all the conditions of Theorem \ref{Thm: Quasi-cyclic of order L/2} are satisfied and hence, $\Gamma(\mathcal{L},g,1,2,2z+1)$ is a quasi-cyclic MTG code over $\mathbb{F}_3$. By Magma \cite{MAGMA}, $\Gamma(\mathcal{L},g,1,2,2z+1)$ is an $[8, 3, 4]$-quasi cyclic code of order $4$ over $\mathbb{F}_3$.
\end{example}
%In Theorem \ref{Thm: Quasi-cyclic of order p}, if $\mathcal{L}=\mathcal{O}_x$ for some $x\in\mathbb{F}_{q^m}$, then we obtain a cyclic twisted Goppa code.
%\begin{theorem}
    %Let $b\in \mathbb{F}_{q^m}^*$ and $\Gamma(\mathcal{L},g,1,t-2,\eta)$ be a twisted Goppa code with parameters as (\ref{pc-mat}), where the defining set $\mathcal{L}=\mathcal{O}_x,$ for some $x\in \mathbb{F}_{q^m}$ and $\mathcal{O}_x$ is an $\mathcal{G}_{1,b}$-orbit. If $g(x)=f(x^p-b^{p-1}x)$ for some $f(x)\in\mathbb{F}_{q^m}$, then the twisted Goppa code $\Gamma(\mathcal{L},g,1,t-2,\eta)$ is a $[p, k, d]$-cyclic code, where $k\ge p-mt$ and $d\ge t$.
%\end{theorem}

One may now follow \cite{Sui2023} and \cite{gaborit2005} to see how to reduce the public key size using quasi-cyclic codes. In the context of reducing the public-key size of McEliece-type cryptosystems, we propose the use of quasi-cyclic MTG codes with $\ell = 1$, where $q$ is a power of an odd prime $p$ and $t = p^{s} + 1$. This choice not only achieves a reduction in key size but also provides resistance to partial key-recovery attacks, as discussed in Section IV-A. 
To achieve NIST post-quantum security Level~5 while simultaneously addressing the public key size constraint, we select the parameters $p=3$, $m=2\cdot 9=18$, $t=1+3^{5}=244$, and $n=8192$ in the proposed construction. 
For these parameters, let $b \in \mathbb{F}_{p^{2m}}^{*}$ and consider the MTG code $\Gamma(\mathcal{L}, g, 1, t-2, \eta)$. By Theorem~\ref{Thm: Quasi-cyclic of order L/2}, the resulting MTG code is quasi-cyclic of order $4096$ and has dimension at least $n - mt = 3800$. 
Moreover, the associated parity-check matrix can be expressed in a compact block form, analogous to the representation described in \cite{Sui2023}, which enables a substantial reduction in the public key size without compromising the targeted security level.

It is to note that quasi-cyclic variants of Goppa code-based public-key encryption have historically been shown to be vulnerable to structural key-recovery attacks exploiting large automorphism groups, most notably through folding techniques and algebraic cryptanalysis \cite{FOPT2010,FOPT2014,Otmani2011}. These attacks crucially rely on the existence of non-trivial permutation symmetries that allow the reduction of the public code to a significantly smaller folded code whose algebraic structure can be efficiently recovered. In contrast, the QC MTG code considered in this work is explicitly designed to break such symmetry. 
As a result, the public code does not admit a meaningful reduction to a smaller equivalent Goppa code, and the algebraic modeling techniques underlying the Faug\`{e}re–Otmani–Perret–Tillich attacks are no longer applicable.

From a security standpoint, the absence of exploitable automorphisms ensures that the public key cannot be compressed into a lower-dimensional algebraic instance, and key recovery reduces to generic decoding or codeword-finding attacks against random-looking Goppa codes. 
This places the security of the proposed scheme in line with that of the original McEliece cryptosystem \cite{McEliece1978}.

Although several works have demonstrated the existence of algebraic or statistical distinguishers for high-rate alternant and Goppa codes, these results are inherently confined to the very high-rate regime and rely on specific algebraic properties of alternant constructions \cite{Faugere2013,Bernstein2008}. In particular, known distinguishers exploit deviations in the Schur square or higher-order products that become detectable only when the code rate is close to one. For codes of moderate rate, such as the proposed construction with rate $3800/8192 \approx 0.46$, no such distinguishing techniques are known, even for structured families, and generic random linear codes of comparable parameters are believed to be indistinguishable from truly random instances. Moreover, since the proposed scheme does not rely on Goppa or alternant structure, it is not subject to the algebraic relations underlying high-rate McEliece distinguishers. Consequently, the security of the construction is governed solely by the hardness of generic decoding, for which no structural or rate-based attacks are known in this parameter range \cite{Finiasz2009,Peters2010}.

\section{CONCLUSIONS}
In this paper, we established the theoretical framework of \emph{multi-twisted Goppa (MTG) codes} by presenting them as the duals of subfield subcodes of generalized multi-twisted Reed–Solomon codes. 
We further provided an explicit, self-contained definition of MTG codes and analyzed their algebraic structure through the construction of corresponding parity-check matrices. 
This formulation enabled a detailed study of their error-correction potential and led to a methodology for designing MTG codes with error-correction capabilities suitable for code-based cryptographic applications. 

Building upon results available in the literature, we also provided a decoding process for MTG codes with a twist applied at an arbitrary position, which extends the decoding for single twisted Goppa codes. 
A rigorous complexity analysis demonstrated that the proposed decoding algorithm remains computationally efficient while preserving decoding accuracy. 
From a cryptographic perspective, we showed that partial key recovery attacks become computationally infeasible when MTG structures are employed, and that the use of quasi-cyclic variants of MTG codes 
(over fields of odd characteristic) provides an effective approach for mitigating public-key size issues without compromising security. 

Future research directions include the design and analysis of efficient decoding algorithms for multi-twisted Goppa codes with more than one twist, a deeper examination of their algebraic and combinatorial properties, 
and a comprehensive study of the resistance of MTG-based cryptosystems against information-set decoding and structural attacks. 
Further exploration of their quasi-cyclic realizations may also yield valuable insights for achieving compact and secure post-quantum code-based cryptosystems.

\section*{Declarations}
\subsection*{Conflict of Interest}
All authors declare that they have no conflict of interest.

\subsection*{Acknowledgments}
The second author acknowledges the support of the Council of Scientific and Industrial Research (CSIR) India, under grant no. 09/0086(13310)/2022-EMR-I.

\bibliographystyle{IEEEtran}
\bibliography{biblio}

\end{document}